\documentclass[envcountsame]{llncs}

\usepackage{amssymb}
\usepackage{amsmath}
\usepackage{mathrsfs}
\usepackage[algoruled]{algorithm2e}
\usepackage[all]{xy}
\usepackage{url}
\usepackage{tabularx}

\newcommand{\Nat}{\bbbn_0}

\newcommand{\pows}[1]{\mathscr{P}(#1)}
\newcommand{\exptime}{\textsc{exptime}}
\newcommand{\nexptime}{\textsc{nexptime}}

\newcommand{\bigo}{\mathcal{O}}
\newcommand{\opt}[1]{#1^\bot}
\newcommand{\updfkt}[3]{#1[#2\mapsto#3]}
\newcommand{\seq}[1]{\mathrm{Seq}(#1)}

\newcommand{\ximp}{\;\Rightarrow\;}
\newcommand{\xand}{\;\&\;}
\newcommand{\xor}{\text{ or }}

\newcommand{\xnot}{\text{not }}

\newcommand{\fml}{\mathrm{Fml}}
\newcommand{\afml}{\mathrm{AFml}}
\newcommand{\prg}{\mathrm{Prg}}
\newcommand{\aprg}{\mathrm{APrg}}
\newcommand{\lprg}{\mathrm{LPrg}}
\newcommand{\feal}{$\pea{\mathrm{lp}}{}$}
\newcommand{\fmlev}{\mathrm{Ev}}

\newcommand{\pnot}[1]{\lnot #1}
\newcommand{\pneg}[1]{\sim\!#1}
\newcommand{\pand}[2]{#1 \land #2}
\newcommand{\por}[2]{#1 \lor #2}
\newcommand{\pea}[2]{\langle#1\rangle #2}
\newcommand{\paa}[2]{[#1] #2}
\newcommand{\pcv}[1]{#1^-}
\newcommand{\psp}[2]{#1;#2}
\newcommand{\pup}[2]{#1 \cup #2}
\newcommand{\prp}[1]{#1*}
\newcommand{\pip}[1]{#1?}
\newcommand{\pimp}[2]{#1 \to #2}
\newcommand{\peq}[2]{#1 \leftrightarrow #2}

\newcommand{\evv}{\mathcal{V}}

\newcommand{\psr}[2]{#1 \Vdash #2}
\newcommand{\npsr}[2]{#1 \nVdash #2}
\newcommand{\prel}[3]{#1 \, #2 \, #3}
\newcommand{\pnnf}{\mathop{\mathrm{nnf}}}
\newcommand{\pnnp}[1]{#1^\smallsmile}
\newcommand{\pcl}{\mathop{\mathrm{cl}}}
\newcommand{\ppre}{\mathop{\mathrm{pre}}}
\newcommand{\pzz}{\mathrel{\rightsquigarrow}}
\newcommand{\extr}{\mathop{\mathrm{ex}}}
\newcommand{\mchn}{\mathop{\mathrm{chn}}}

\SetKwInput{Input}{Input}
\SetKwInput{Output}{Output}
\SetKwInput{Remark}{Remark}
\newcommand{\bcmt}{$(\!*$}
\newcommand{\ecmt}{$*\!)$}
\newcommand{\prcdr}{\textbf{Procedure}}
\newcommand{\issat}{\texttt{is-sat}}
\newcommand{\expand}{\texttt{expand}}
\newcommand{\detsts}{\texttt{det-status}}
\newcommand{\detstsb}{\texttt{det-sts-$\beta$}}
\newcommand{\detstss}{\texttt{det-sts-state}}
\newcommand{\detstsd}{\texttt{det-sts-spl}}
\newcommand{\detstscint}{\texttt{det-prs-child}}
\newcommand{\filter}{\texttt{filter}}

\newcommand{\fset}{\Gamma}
\newcommand{\ann}{\mathop{\mathrm{ann}}}
\newcommand{\pst}{\mathrm{pst}}
\newcommand{\pap}{\mathrm{ppr}}
\newcommand{\tms}{\mathrm{idx}}
\newcommand{\sts}{\mathrm{sts}}
\newcommand{\stsv}{\mathfrak{S}}
\newcommand{\unexp}{\mathtt{unexp}}
\newcommand{\undef}{\mathtt{undef}}
\newcommand{\closed}{\mathop{\mathtt{closed}}}
\newcommand{\open}{\mathop{\mathtt{open}}}
\newcommand{\prs}{\mathop{\mathrm{prs}}}
\newcommand{\alt}{\mathrm{alt}}
\newcommand{\mrk}{\mathrm{cs}}
\newcommand{\prsval}{\Lambda}
\newcommand{\rt}{\mathrm{rt}}

\newcommand{\create}{\mathop{\mbox{create-new-node}}}
\newcommand{\chld}{\mathop{\mathrm{getChild}}}
\newcommand{\eann}{\mathop{\mathrm{ann^\bot}}}
\newcommand{\eprs}{\mathop{\mathrm{prs^\bot}}}
\newcommand{\annchn}{\mathop{\mathrm{defer}}}
\newcommand{\prsreach}{\mathop{\mathrm{reach}}}

\newcommand{\nords}{\mathrel{\sqsubset}}

\newcommand{\gm}{G_{\mathrm{p}}}
\newcommand{\fchn}{\sigma}
\newcommand{\nchn}{\omega}
\newcommand{\ini}[1]{#1_{\mathrm{i}}}
\newcommand{\ptya}{\mathop{\mathfrak{P}}}

\begin{document}

\pagestyle{plain}

\title{Optimal and Cut-free Tableaux for Propositional Dynamic Logic with Converse}

\author{Rajeev Gor{\'e}\inst{1} \and Florian Widmann\inst{1,2}}

\institute{
  Logic and Computation Group$^1$ and NICTA$^2$\thanks{NICTA is funded by the
    Australian Government's Department of Communications, Information
    Technology and the Arts and the Australian Research Council
    through Backing Australia's Ability and the ICT Centre of
    Excellence program.},
  The Australian National University\\
  Canberra, ACT 0200, Australia,
  \email{$\{$Rajeev.Gore,Florian.Widmann$\}$@anu.edu.au}
}

\maketitle

\begin{abstract}
  We give an optimal (\exptime{}), sound and complete tableau-based algorithm
  for deciding satisfiability for propositional dynamic logic with converse (CPDL)
  which does not require the use of analytic cut.  
  Our main contribution is a sound method
  to combine our previous optimal method for tracking least fix-points in PDL
  with our previous optimal method for handling converse in the description logic~$ALCI$.
  The extension is non-trivial as the two methods cannot be combined naively.
  We give sufficient details to enable an implementation by others.
  Our OCaml implementation seems to be the first theorem prover for CPDL.
\end{abstract}

\section{Introduction}

Propositional dynamic logic (PDL) is an important logic for reasoning about programs.
Its formulae consist of traditional Boolean formulae
plus ``action modalities'' built from a finite set of atomic programs
using sequential composition~$(\psp{}{})$, non-deterministic choice~$(\pup{}{})$,
repetition~$(\prp{})$, and test~$(\pip{})$.
The logic CPDL is obtained by adding  converse ($\pcv{}$),
which allows us to reason about previous actions.
The satisfiability problem for CPDL is \exptime{}-complete~\cite{vardi85:tamin_conver}.

De Giacomo and Massacci~\cite{de-giacomo-massacci-tableaux-converse-pdl}
give an \nexptime{} tableau algorithm for deciding CPDL-satisfiability,
and discuss ways to obtain optimality,
but do not give an actual \exptime{} algorithm.
The tableau method of Nguyen and Sza{\l}as~\cite{nguyen-szalas-cpdl} is optimal.
Neither method has been implemented,
and since both require an explicit analytic cut rule,
it is not at all obvious that they can be implemented efficiently.
Optimal game-theoretic methods for fix-point logics~\cite{lange-stirling-focus-games}
can be adapted to handle CPDL~\cite{lange03:satis_compl_conver_pdl_replay}
but involve significant non-determinism.
Optimal automata-based methods~\cite{vardi-wolper-automata}
for fix-point logics are still in their infancy
because good optimisations are not known.
We know of no resolution methods for CPDL.

We give an optimal tableau method for deciding CPDL-satisfiability
which does not rely on a cut rule.
Our main contribution is a sound method to combine our method
for tracking and detecting unfulfilled eventualities
as early as possible in PDL~\cite{gore-widmann-pdl-cade09}
with our method for handling converse for $ALCI$~\cite{gore-widmann-alci-tableaux09}.
The extension is non-trivial as the two methods cannot be combined naively.

We present a mixture of pseudo code and tableau rules
rather than a set of traditional tableau rules to enable easy implementation by others.
Our unoptimised OCaml implementation appears to be the first automated theorem prover for CPDL
(\url{http://rsise.anu.edu.au/~rpg/CPDLTabProver/}).
The proofs can be found in the appendix.

\section{Syntactic Preliminaries}
\label{secd:syntax-semantics}

\begin{definition}
  Let~$\afml$ and~$\aprg$ be two disjoint and countably infinite sets
  of propositional variables and \emph{atomic programs}, respectively.
  The set~$\lprg$ of \emph{literal programs} is defined
  as $\lprg := \aprg \cup \{ \pcv{a} \mid a \in \aprg \}$.
  The set~$\fml$ of all formulae and the set~$\prg$ of all \emph{programs}
  are defined mutually inductively as follows where~$p \in \afml$ and~$l \in \lprg$:

  \begin{tabular}[c]{ccccccccccccccc}
    $\fml$ \ \ \ & $\varphi$ & ::=
    & $p$ & $\mid$
    & $\pnot{\varphi}$ & $\mid$
    & $\pand{\varphi}{\varphi}$ & $\mid$
    & $\por{\varphi}{\varphi}$ & $\mid$
    & $\pea{\gamma}{\varphi}$ & $\mid$ 
    & $\paa{\gamma}{\varphi}$
    \\
    $\prg$ \ \ \ & $\gamma$ & ::=
    & $l$ & $\mid$
    & $\psp{\gamma}{\gamma}$ & $\mid$
    & $\pup{\gamma}{\gamma}$ & $\mid$
    & $\prp{\gamma}$ & $\mid$
    & $\pip{\varphi} \enspace.$
  \end{tabular}

  \noindent{} A \feal-formula is a formula~$\pea{\gamma}{\varphi}$
  where~$\gamma \in \lprg$ is a \underline{\emph{l}}iteral \underline{\emph{p}}rogram.
\end{definition}

Implication~($\pimp{}{}$) and equivalence~($\peq{}{}$)
are not part of the core language but can be defined as usual.
In the rest of the paper, let~$p \in \afml$ and~$l \in \lprg$.

We omit the semantics as it is a straightforward extension of PDL~\cite{gore-widmann-pdl-cade09}
and write~$\psr{M, w}{\varphi}$ if~$\varphi \in \fml$ holds in the world~$w \in W$ of the model~$M$.

\begin{definition}
  For a literal program~$l \in \lprg$, we define~$\pnnp{l}$ as~$a$
  if~$l$ is of the form~$\pcv{a}$, and as~$\pcv{l}$ otherwise.
  A formula~$\varphi \in \fml$ is in \emph{negation normal form}
  if the symbol~$\pnot{}$ appears only directly before propositional variables.
  For every~$\varphi \in \fml$, we can obtain a formula~$\pnnf(\varphi)$ in negation normal form
  by pushing negations inward such that~$\peq{\varphi}{\pnnf{\varphi}}$ is valid.
  We define~$\pneg{\varphi} := \pnnf(\pnot{\varphi})$.
\end{definition}

We categorise formulae as $\alpha$- or $\beta$-formulae as shown in Table~\ref{tabd:alphabeta}
so that the formulae of the form~$\peq{\alpha}{\pand{\alpha_1}{\alpha_2}}$ and~$\peq{\beta}{\por{\beta_1}{\beta_2}}$ are valid.
\begin{table}[t]
  \caption{Smullyan's $\alpha$- and $\beta$-notation to classify formulae}
  \label{tabd:alphabeta}
  \centering
  \begin{tabular}{|c|c|c|c|c|c|c|}
    \hline
    $\alpha$ 
    & $\pand{\varphi}{\psi}$ 
    & $\paa{\pup{\gamma}{\delta}}{\varphi}$ 
    & $\paa{\prp{\gamma}}{\varphi}$ 
    & $\pea{\pip{\psi}}{\varphi}$ 
    & $\pea{\psp{\gamma}{\delta}}{\varphi}$ 
    & $\paa{\psp{\gamma}{\delta}}{\varphi}$ 
    \\ \hline
    $\alpha_1$ 
    & $\varphi$ 
    & $\paa{\gamma}{\varphi}$ 
    & $\varphi$ 
    & $\varphi$ 
    & $\pea{\gamma}{\pea{\delta}{\varphi}}$ 
    & $\paa{\gamma}{\paa{\delta}{\varphi}}$ 
    \\ \hline
    $\alpha_2$
    & $\psi$
    & $\paa{\delta}{\varphi}$
    & $\paa{\gamma}{\paa{\prp{\gamma}}{\varphi}}$
    & $\psi$
    & 
    & 
    \\ \hline
  \end{tabular}
  \begin{tabular}{|c|c|c|c|c|}
    \hline
    $\beta$ 
    & $\por{\varphi}{\psi}$ 
    & $\pea{\pup{\gamma}{\delta}}{\varphi}$ 
    & $\pea{\prp{\gamma}}{\varphi}$ 
    & $\paa{\pip{\psi}}{\varphi}$ 
    \\ \hline
    $\beta_1$ 
    & $\varphi$ 
    & $\pea{\gamma}{\varphi}$ 
    & $\varphi$ 
    & $\varphi$ 
    \\ \hline
    $\beta_2$
    & $\psi$
    & $\pea{\delta}{\varphi}$
    & $\pea{\gamma}{\pea{\prp{\gamma}}{\varphi}}$
    & $\pneg{\psi}$
    \\ \hline
  \end{tabular}
\end{table}
An \emph{eventuality} is a formula of the form~$\pea{\gamma_1}{\dotsc \pea{\gamma_k}{\pea{\prp{\gamma}}{\varphi}}}$,
and~$\fmlev$ is the set of all eventualities. 
Using Table~\ref{tabd:alphabeta}, the binary relation~``$\pzz$''
relates a $\pea{}{}$-formulae~$\alpha$ (respectively~$\beta$),
to its reduction~$\alpha_1$ (respectively~$\beta_1$ and~$\beta_2$).
See~\cite[Def.~7]{gore-widmann-pdl-cade09} for their formal definitions.

\section{An Overview of our Algorithm}
\label{sec:algorithm}

Our algorithm builds an and-or graph~$G$
by repeatedly applying four rules (see Table~\ref{tabd:rules})
to try to build a model for a given~$\phi$ in negation normal form. 
Each node~$x$ carries a formula set~$\Gamma_x$, a status~$\sts_x$,
and other fields to be described shortly. 
Rule~1 applies the usual expansion rules to a node to create its children.
These expansion rules capture the semantics of CPDL.
We use Smullyan's $\alpha/\beta$-rule notation for classifying rules and nodes.
As usual, a node~$x$ is a (``saturated'') \emph{state}
if no $\alpha/\beta$-rule can be applied to it.
If~$x$ is a state then for each~$\pea{l}{\xi}$ in~$\Gamma_x$,
we create a node~$y$ with~$\Gamma_y = \{ \xi \} \cup \Delta$, where~$\Delta = \{\psi \mid \paa{l}{\psi} \in \Gamma_x\}$,
and add an edge from~$x$ to~$y$ labelled with~$\pea{l}{\xi}$
to record that~$y$ is an $l$-successor of~$x$. 

If~$\Gamma_x$ contains an obvious contradiction during
expansion, its status becomes ``closed'', which is irrevocable.
Else, at some later stage, Rule~2 determines its status as either ``closed'' or ``open''.
``Open'' nodes contain additional information
which depends on the status of other nodes.
Hence, if a node changes its status, it might affect the status of another (``open'') node.
If the stored status of a node does not match its current status,
the node is no longer \emph{up-to-date}.
Rule~3, which may be applied multiple times to the same node,
ensures that ``open'' nodes are kept up-to-date
by recomputing their status if necessary.
Finally, Rule~4 detects eventualities which are impossible to fulfil
and closes nodes which contain them. 
We first describe the various important components of our algorithm separately.

\paragraph{Global State Caching.}
For optimality, the graph~$G$ never contains two state nodes
which carry the same set of formulae~\cite{gore-widmann-alci-tableaux09}.
However, there may be multiple non-states which carry the same set of formulae.
That is, a non-state node~$x$ carrying~$\fset$
which appears while saturating a child~$y$ of a state~$z$ is unique to~$y$.
If a node carrying~$\fset$ is required in some other saturation phase,
a new node carrying~$\fset$ is created.
Hence the nodes of two saturation phases are distinct.

\paragraph{Converse.}
Suppose state~$y$ is a descendant of an $l$-successor of a state~$x$,
with no intervening states. 
Call~$x$ the parent state of~$y$ since all intervening nodes are not states.
We require that $\{ \psi \mid \paa{\pcv{l}}{\psi} \in \Gamma_y \} \subseteq \Gamma_x$,
since~$y$ is then compatible with being a $l$-successor of~$x$
in the putative model under construction.
If some~$\paa{\pcv{l}}{\psi} \in \Gamma_y$ has~$\psi \notin \Gamma_x$ then~$x$ is ``too small'',
and must be ``restarted'' as an alternative node~$x^{+}$ containing all such~$\psi$.
If any such~$\psi$ is a complex formula to which an $\alpha/\beta$-rule is applicable
then~$x^{+}$ is not a state and may have to be ``saturated'' further.
The job of creating these alternatives is done by  
\emph{special nodes}~\cite{gore-widmann-alci-tableaux09}.
Each special node monitors a state and creates the alternatives when needed. 

\paragraph{Detecting Fulfilled and Unfulfilled Eventualities.}
Suppose the current node~$x$ contains an eventuality~$e_x$. 
There are three possibilities.
The first is that~$e_x$ can be fulfilled in the part of the graph which is ``older'' than~$x$.
Else, it may be possible to reach a node~$z$ in the parts of the graph ``newer'' than~$x$
such that~$z$ contains a reduction~$e_z$ of~$e_x$.
Since this ``newer'' part of the graph is not fully explored yet,
future expansions may enable us to fulfil~$e_x$ via~$z$,
so the pair~$(z, e_z)$ is a ``potential rescuer'' of~$e_x$.
The only remaining case is that~$e_x$ cannot be fulfilled in the ``older'' part of the graph,
and has no potential rescuers.
Thus future expansions of the graph cannot possibly help to fulfil~$e_x$
since it cannot reach these ``newer'' parts of the future graph.
In this case~$x$ can be ``closed''. 
The technical machinery to maintain this information for PDL is
from~\cite{gore-widmann-pdl-cade09}.
However, the presence of ``converse'' and the resulting need for alternative nodes
requires a more elaborate scheme for CPDL.

\section{The Algorithm}

Our algorithm builds a directed graph~$G$
consisting of nodes and directed edges.
We first explain the structure of~$G$ in more detail.
\begin{definition}
  Let~$X$ and~$Y$ be sets.
  We define~$\opt{X} := X \uplus \{ \bot \}$
  where~$\bot$ indicates the undefined value and~$\uplus$ is the disjoint union.
  If~$f: X \to Y$ is a function and~$x \in X$ and~$y \in Y$
  then the function~$\updfkt{f}{x}{y}: X \to Y$ is defined as
  $\updfkt{f}{x}{y}(x') := y$ if~$x' = x$
  and $\updfkt{f}{x}{y}(x') := f(x')$ if~$x' \not= x$.
\end{definition}
\begin{definition}
  \label{def:nodes}
  Let~$G = (V, E)$ be a graph
  where~$V$ is a set of nodes and~$E$ is a set of directed edges.
  Each node~$x \in V$ has six attributes:
  $\fset_x \subseteq \fml$, $\ann_x: \fmlev \to \opt{\fml}$,
  $\pst_x \in \opt{V}$, $\pap_x \in \opt{\lprg}$,
  $\tms_x \in \opt{Nat}$, and~$\sts_x \in \stsv$
  where~$\stsv := \{ \unexp, \undef \} \cup \{ \closed(\alt) \mid \alt \subseteq \pows{\fml} \} \cup
  \{ \open(\prs,\alt) \mid \prs: \fmlev \to \opt{(\pows{V \times \fmlev})} \xand \alt \subseteq \pows{\fml} \}$.
  Each directed edge~$e \in E$ is labelled with a label~$l_e \in \opt{(\fml \cup \pows{\fml} \cup \{ \mrk \})}$
  where~$\mrk$ is just a constant.
\end{definition}

All attributes of a node~$x \in V$ are initially set at the creation of~$x$,
possibly with the value~$\bot$ (if allowed).
Only the attributes~$\tms_x$ and~$\sts_x$ are changed at a later time.
We use the function~$\create(\fset, \ann, \pst, \pap, \tms, \sts)$
to create a new node and initialise its attributes in the obvious way.

The finite set~$\fset_x$ contains the formulae which are assigned to~$x$.
The attribute~$\ann_x$ is defined for the eventualities in~$\fset_x$ at most.
If~$\ann_x(\varphi) = \varphi'$ then~$\varphi' \in \fset_x$ and~$\varphi \pzz \varphi'$.
The intuitive meaning is that~$\varphi$ has already been ``reduced'' to~$\varphi'$ in~$x$.
For a state (as defined below) we always have that~$\ann_x$ is undefined everywhere
since we do not need the attribute for states.

The node~$x$ is called a \emph{state} iff both attributes~$\pst_x$ and~$\pap_x$ are undefined.
For all other nodes, the attribute~$\pst_x$ identifies the,
as we will ensure, unique ancestor~$p \in V$ of~$x$
such that~$p$ is a state and there is no other state between~$p$ and~$x$ in~$G$.
We call~$p$ the \emph{\underline{p}arent \underline{st}ate} of~$x$.
The creation of the child of~$p$ which lies on the path from~$p$ to~$x$ (it could be~$x$)
was caused by a \feal-formula~$\pea{l}{\varphi}$ in~$\fset_p$.
The literal program~$l$ which we call the \emph{\underline{p}arent \underline{pr}ogram} of~$x$
is stored in~$\pap_x$.
Hence, for nodes which are not states, both~$\pst_x$ and~$\pap_x$ are defined.

The attribute~$\sts_x$ describes the \emph{status} of~$x$.
Unlike the attributes described so far, its value may be modified several times.
The value~$\unexp$, which is the initial value of each node,
indicates that the node has not yet been expanded.
When a node is expanded, its status becomes either~$\closed(\cdot)$
if it contains an immediate contradiction,
or~$\undef$ to indicate that the node has been expanded
but that its ``real'' status is to be determined.
Eventually, the status of each node is set to either~$\closed(\cdot)$ or~$\open(\cdot, \cdot)$.
If the status is~$\open(\cdot, \cdot)$, it might be modified several times later on,
either to~$\closed(\cdot)$ or to~$\open(\cdot, \cdot)$ (with different arguments),
but once it becomes $\closed(\cdot)$, it will never change again.

We call a node \emph{undefined} if its status is~$\unexp$ or~$\undef$
and \emph{defined} otherwise.
Hence a node is undefined initially,
becomes defined eventually, and then never becomes undefined again.
Furthermore, we call~$x$ \emph{closed} iff its status is~$\closed(\alt)$
for some~$\alt \subseteq \pows{\fml}$.
In this case, we define~$\alt_x := \alt$.
We call~$x$ \emph{open} iff its status is~$\open(\prs,\alt)$
for some~$\prs: \fmlev \to \opt{(\pows{V \times \fmlev})}$ and some~$\alt \subseteq \pows{\fml}$.
In this case, we define~$\prs_x := \prs$ and~$\alt_x := \alt$.
To avoid some clumsy case distinctions,
we define~$\alt_x := \emptyset$ if~$x$ is undefined.

The value~$\closed(\alt)$ indicates
that the node is ``useless'' for building an interpretation
because it is either unsatisfiable or ``too small''.
In the latter case,
the set~$\alt$ of \emph{alternative sets} contains information about missing formulae.
Finally, the value~$\open(\prs, \alt)$ indicates
that there is still hope that~$x$ is ``useful''
and the function~$\prs_x$ contains information about each eventuality~$e_x \in \fset_x$
as explained in the overview.
Although~$x$ itself may be useful,
we need its alternative sets in case it becomes closed later on.
Hence it also has a set of alternative sets.

The attribute~$\tms_x$ serves as a time stamp.
It is set to~$\bot$ at creation time of~$x$
and becomes defined when~$x$ becomes defined.
When this happens, the value of~$\tms_x$ is set
such that~$\tms_x > \tms_y$ for all nodes~$y$ which became defined earlier than~$x$.
We define~$y \nords x$ iff~$\tms_y \not= \bot$ and either~$\tms_x = \bot$ or~$\tms_y < \tms_x$.
Note that~$y \nords x$ depends on the current state of the graph.
However, once~$y \nords x$ holds, it will do so for the rest of the time.

To track eventualities, we label an edge between a state and one of its children
by the \feal-formula~$\pea{l}{\varphi}$ which creates this child.
Additionally, we label edges from special nodes (see overview)
to their \underline{c}orresponding \underline{s}tates with the marker~$\mrk$.
We also label edges from special nodes to its alternative nodes
with the corresponding alternative set.

\begin{definition}
  Let~$\eann: \fmlev \to \opt{\fml}$ and~$\eprs: \fmlev \to \opt{(\pows{V \times \fmlev})}$
  be the functions which are undefined everywhere.
  For a node~$x \in V$ and a label~$l \in \fml \cup \pows{\fml} \cup \{ \mrk \}$,
  let~$\chld(x, l)$ be the node~$y \in V$
  such that there exists an edge~$e \in E$ from~$x$ to~$y$ with~$l_e = l$.
  If~$y$ does not exists or is not unique, let the result be~$\bot$.
  For a function~$\prs: \fmlev \to \opt{(\pows{V \times \fmlev})}$,
  a node~$x \in V$, and an eventuality~$\varphi \in \fmlev$,
  we define the set~$\prsreach(\prs, x, \varphi)$ of eventualities as follows:
  \[ \prsreach(\prs, x, \varphi) \begin{array}[t]{l}
    := \Big\{ \psi \in \fmlev \mid 
    \exists k \in \Nat.\: \exists \varphi_0, \dotsc, \varphi_k \in \fmlev.\: \Big( \psi = \varphi_k \xand\\
    (x, \varphi_0) \in \prs(\varphi) \xand 
    \forall i \in \{ 0, \dotsc, k-1 \}.\: (x, \varphi_{i+1}) \in \prs(\varphi_i) \Big) \Big\} \enspace.
  \end{array} \]
  The function~$\annchn: V \times \fmlev \to \opt{\fml}$ is defined as follows:
  \begin{displaymath}
    \annchn(x, \varphi) := 
    \left\{
      \begin{array}{ll}
        \psi & \text{ if }
        \begin{array}[t]{l}
          \exists k \in \Nat.\: \exists \varphi_0, \dotsc, \varphi_k \in \fml.\: \Big( \varphi_0 = \varphi \xand \varphi_k = \psi \xand\\
          \forall i \in \{ 0, \dotsc, k-1 \}.\: \big( \varphi_i \in \fmlev \xand \ann_x(\varphi_i) = \varphi_{i+1} \big) \xand\\
          \big( \varphi_k \notin \fmlev \xor \ann_x(\varphi_k) = \bot \big) \Big)
        \end{array}\\
        \bot & \text{ otherwise.}
      \end{array}
    \right .
  \end{displaymath}
\end{definition}

The function~$\chld(x, l)$ retrieves a particular child of~$x$.
It is easy to see that, during the algorithm,
the child is always unique if it exists.

Intuitively, the function~$\prsreach(\prs, x, \varphi)$ computes all eventualities
which can be ``reached'' from~$\varphi$ inside~$x$ according to~$\prs$.
If a potential rescuer~$(x, \psi)$ is contained in~$\prs(\varphi)$,
the potential rescuers of~$\psi$ are somehow relevant for~$\varphi$ at~$x$.
Therefore~$\psi$ itself is relevant for~$\varphi$ at~$x$.
The function~$\prsreach(\prs, x, \varphi)$ computes exactly the transitive closure
of this relevance relation.

Intuitively, the function~$\annchn(x, \varphi)$ follows the ``$\ann_x$-chain''.
That is, it computes~$\varphi_1 := \ann_x(\varphi)$, $\varphi_2 := \ann_x(\varphi_1)$, and so on.
There are two possible outcomes.
The first outcome is that we eventually encounter a~$\varphi_k$
which is either not an eventuality or has~$\ann_x(\varphi_k) = \bot$.
Consequently, we cannot follow the ``$\ann_x$-chain'' any more.
In this case we stop and return~$\annchn(x, \varphi) := \varphi_k$.
The second outcome is that we can follow the ``$\ann_x$-chain'' indefinitely.
Then, as~$\fset_x$ is finite, there must exist a cycle~$\varphi_0, \dotsc, \varphi_n, \varphi_0$ of eventualities
such that~$\ann_x(\varphi_i) = \varphi_{i+1}$ for all~$0 \leq i < n$, and~$\ann_x(\varphi_n) = \varphi_0$.
In this case we say that~$x$ (or~$\fset_x$) contains an \emph{``at a world'' cycle}
and return~$\annchn(x, \varphi) := \bot$.

Next we comment on all procedures given in pseudocode.

\begin{procedure}[t]
  \caption{\issat{}($\phi$) for testing whether a formula~$\phi$ is satisfiable}
  \dontprintsemicolon
  \SetVline
  \Input{a formula~$\phi \in \fml$ in negation normal form}
  \Output{true iff~$\phi$ is satisfiable}
  \BlankLine
  $G :=$ a new empty graph; \quad $\tms := 1$ \;
  let~$d \in \aprg$ be a dummy atomic program which does not occur in~$\phi$ \;
  $\rt := \create(\{ \pea{d}{\phi} \}, \eann, \bot, \bot, \bot, \unexp)$ \;
  insert~$\rt$ in~$G$ \;
  \While{one of the rules in Table~\ref{tabd:rules} is applicable}{
    apply any one of the applicable rules in Table~\ref{tabd:rules} \;
  }
  \lIf{$\sts_\rt = \open(\cdot,\cdot)$}{\Return true} \lElse{\Return false} \;
\end{procedure}
\begin{table}[t]
  \caption{Rules used in the procedure~\issat{}}
  \label{tabd:rules}
  \centering
  \begin{tabular}{|ll|}
    \hline
    Rule~1: & Some node~$x$ has not been expanded yet.\\
    Condition: & $\exists x \in V.\: \sts_x = \unexp$\\
    Action: & \expand($x$)\\
    \hline
    Rule~2: & The status of some node~$x$ is still undefined.\\
    Condition: & $\exists x \in V.\: \sts_x = \undef$\\
    Action: & $\sts_x :=$ \detsts($x$) \xand $\tms_x := \tms$ \xand $\tms := \tms + 1$\\
    \hline
    Rule~3: & Some open node~$x$ is not up-to-date.\\
    Condition: & $\exists x \in V.\: \open(\cdot,\cdot) = \sts_x \not= $ \detsts($x$)\\
    Action: & $\sts_x :=$ \detsts($x$)\\
    \hline
    Rule~4: & All nodes are up-to-date, and some~$x$ has an unfulfilled eventuality~$\varphi$.\\
    Condition: & Rule~3 is not applicable and\\
    & $\exists x \in V.\: \sts_x = \open(\prs_x, \alt_x) \xand \exists \varphi \in \fmlev \cap \fset_x.\: \prs_x(\varphi) = \emptyset$\\
    Action: & $\sts_x := \closed(\alt_x)$\\
    \hline
  \end{tabular}
\end{table}

\noindent\prcdr{}~\issat{}($\phi$) is invoked to determine
whether a formula~$\phi \in \fml$ in negation normal form is satisfiable.
It creates a root node~$\rt$
and initialises the graph~$G$ to contain only~$\rt$.
The dummy program~$d$ is used to make~$\rt$ a state
so that each node in~$G$ which is not a state has a parent state.
The global variable~$\tms$ is used to set the time stamps of the nodes accordingly.

While at least one of the rules in Table~\ref{tabd:rules} is applicable,
that is its condition is true,
the algorithm applies any applicable rule.
If no rules are applicable,
the algorithm returns satisfiable iff~$\rt$ is open.

Rule~1 picks an unexpanded node and expands it.
Rule~2 picks an expanded but undefined node and computes its (initial) status.
It also sets the correct time stamp.
Rule~3 picks an open node whose status has changed and recomputes its status.
Its meaning is, that if we compute \detsts($x$) on the current graph
then its result is different from the value in~$\sts_x$,
and consequently, we update~$\sts_x$ accordingly.
Rule~4 is only applicable if all nodes are up-to-date.
It picks an open node containing an eventuality~$\varphi$
which is currently not fulfilled in the graph
and which does not have any potential rescuers either.
As this indicates that~$\varphi$ can never be fulfilled,
the node is closed.

This description leaves several questions open, most notably:
``How do we check efficiently whether Rule~3 is applicable?'' and
``Which rule should be taken if several rules are applicable?''.
We address these issues in Section~\ref{secd:strategy}.

\begin{procedure}[t]
  \caption{\expand{}($x$) for expanding a node~$x$}
  \dontprintsemicolon
  \SetVline
  \Input{a node~$x \in V$ with~$\sts_x = \unexp$}
  \BlankLine
  \If{$\exists\varphi \in \fset_x.\: \pneg{\varphi} \in \fset_x \xor (\varphi \in \fmlev \xand \annchn(x, \varphi) = \bot)$}{
    $\tms_x := \tms$; \quad $\tms := \tms + 1$; \quad  $\sts_x := \closed(\emptyset)$ \;
  }
  \Else(\bcmt{} $x$ does not contain a contradiction \ecmt{}){
    $\sts_x := \undef$ \;
    \uIf(\bcmt{} $x$ is a state \ecmt{}){$\pst_x = \bot$}{
      let~$\pea{l_1}{\varphi_1}, \dotsc , \pea{l_k}{\varphi_k}$ be all of the \feal-formulae in~$\fset_x$ \;
      \For{$i \longleftarrow 1$ \KwTo $k$}{
        $\fset_i := \{ \varphi_i \} \cup \{ \psi \mid \paa{l_i}{\psi} \in \fset_x \}$ \;
        $y_i := \create(\fset_i, \eann, x, l_i, \bot, \unexp)$ \;
        insert~$y_i$, and an edge from~$x$ to~$y_i$ labelled with~$\pea{l_i}{\varphi_i}$, into~$G$ \;
      }
    }
    \uElseIf{$\exists\alpha \in \fset_x.\: \{ \alpha_1, \dotsc, \alpha_k \} \not\subseteq \fset_x \xor (\alpha \in \fmlev \xand \ann_x(\alpha) = \bot)$}{
      $\fset := \fset_x \cup \{ \alpha_1, \dotsc, \alpha_k \}$ \;
      $\ann :=$ \lIf{$\alpha \in \fmlev$}{$\updfkt{\ann_x}{\alpha}{\alpha_1}$} \lElse{$\ann_x$} \;
      $y := \create(\fset, \ann, \pst_x, \pap_x, \bot, \unexp)$ \;
      insert~$y$, and an edge from~$x$ to~$y$, into~$G$ \;
    }
    \uElseIf{$\exists\beta \in \fset_x.\: \{ \beta_1, \beta_2 \} \cap \fset_x = \emptyset \xor (\beta \in \fmlev \xand \ann_x(\beta) = \bot)$}{
      \For{$i \longleftarrow$ 1 \KwTo 2}{
        $\fset_i := \fset_x \cup \{ \beta_i \}$ \;
        $\ann_i :=$ \lIf{$\beta \in \fmlev$}{$\updfkt{\ann_x}{\beta}{\beta_i}$} \lElse{$\ann_x$} \;
        $y_i := \create(\fset_i, \ann_i, \pst_x, \pap_x, \bot, \unexp)$ \;
        insert~$y_i$, and an edge from~$x$ to~$y_i$, into~$G$ \;
      }
    }
    \Else(\bcmt{} $x$ is a special node \ecmt{}){
      \uIf(\bcmt{} state already exists in~$G$ \ecmt{}){
        $\exists y \in V.\: \fset_y = \fset_x \xand \pst_y = \bot$
      }{
        insert an edge from~$x$ to~$y$ labelled with~$\mrk$ into~$G$ \;
      }
      \Else(\bcmt{} state does not exist in~$G$ yet \ecmt{}){
        $y := \create(\fset_x, \eann, \bot, \bot, \bot, \unexp)$ \;
        insert~$y$, and an edge from~$x$ to~$y$ labelled with~$\mrk$, into~$G$ \;
      }
    }
  }
\end{procedure}

\noindent\prcdr{}~\expand{}($x$) expands a node~$x$.
If~$\fset_x$ contains an immediate contradiction or an ``at a world'' cycle
then we close~$x$ and set the time stamp accordingly.
For the other cases, we assume implicitly that~$\fset_x$ does not contain either of these.

If~$x$ is a state, that is~$\pst_x = \bot$,
then we do the following for each \feal-formula~$\pea{l_i}{\varphi_i}$.
We create a new node~$y_i$
whose associated set contains~$\varphi_i$ and all~$\psi$ such that~$\paa{l_i}{\psi} \in \fset_x$.
As none of the eventualities in~$\fset_{y_i}$ is reduced yet, there are no annotations.
The parent state of~$y_i$ is obviously~$x$ and its parent program is~$l_i$.
In order to relate~$y_i$ to~$\pea{l_i}{\varphi_i}$,
we label the edge from~$x$ to~$y_i$ with~$\pea{l_i}{\varphi_i}$.
We call~$y_i$ the \emph{successor} of~$\pea{l_i}{\varphi_i}$.

If~$x$ is not a state and~$\fset_x$ contains an $\alpha$-formula~$\alpha$
whose decompositions are not in~$\fset_x$,
or which is an unannotated eventuality,
we call~$x$ an \emph{$\alpha$-node}.
In this case, we create a new node~$y$
whose associated set is the result of adding all decompositions of~$\alpha$ to~$\fset_x$.
If~$\alpha$ is an eventuality
then~$\ann_y$ extends~$\ann_x$ by mapping~$\alpha$ to~$\alpha_1$.
The parent state and the parent program of~$y$ are inherited from~$x$.
Note that~$\pst_x$ and~$\pap_x$ are defined as~$x$ is not a state.
Also note that~$\fset_y \supsetneq \fset_x$
or~$\alpha$ is an eventuality which is annotated in~$\ann_y$ but not in~$\ann_x$.

If~$x$ is neither a state nor an $\alpha$-node
and~$\fset_x$ contains a $\beta$-formula~$\beta$
such that neither of its immediate subformulae is in~$\fset_x$,
or such that~$\beta$ is an unannotated eventuality,
we call~$x$ a \emph{$\beta$-node}.
For each decomposition~$\beta_i$ we do the following.
We create a new node~$y_i$
whose associated set is the result of adding~$\beta_i$ to~$\fset_x$.
If~$\beta$ is an eventuality
then~$\ann_{y_i}$ extends~$\ann_x$ by mapping~$\alpha$ to~$\beta_i$.
The parent state and the parent program of~$y$ are inherited from~$x$.
Note that~$\pst_x$ and~$\pap_x$ are defined as~$x$ is not a state.
Also note that~$\fset_{y_i} \supsetneq \fset_x$
or~$\beta$ is an eventuality which is annotated in~$\ann_{y_i}$ but not in~$\ann_x$.

If~$x$ is neither a state nor an $\alpha$-node nor a $\beta$-node,
it must be fully saturated and we call it a \emph{special node}.
Intuitively, a special node sits between a saturation phase and a state
and is needed to handle the ``special'' issue arising from converse programs,
as explained in the overview.
Like $\alpha$- and $\beta$-nodes,
special nodes have a unique parent state and a unique parent program.
In this case we check whether there already exists a state~$y$ in~$G$
which has the same set of formulae as the special node.
If such a state~$y$ exists, we link~$x$ to~$y$;
else we create such a state and link~$x$ to it.
In both cases we label the edge with the marker~$\mrk$
since a special node can have several children (see below)
and we want to uniquely identify the $\mrk$-child~$y$ of~$x$.
Note that there is only at most one state for each set of formulae
and that states are always fully saturated since special nodes are.

\begin{procedure}[t]
  \caption{\detsts{}($x$) for determining the status of a node~$x$}
  \dontprintsemicolon
  \SetVline
  \Input{a node~$x \in V$ with~$\unexp \not= \sts_x \not= \closed(\cdot)$}
  \BlankLine
  \lIf{$x$ is an $\alpha$-or a $\beta$-node}{$\sts_x :=$ \detstsb($x$)} \;
  \lElseIf{$x$ is a state}{$\sts_x :=$ \detstss($x$)} \;
  \Else(\bcmt{} $x$ is a special node, in particular~$\pst_x \not= \bot \not= \pap_x$ \ecmt{}){
    $\fset_\alt := \{ \varphi \mid \paa{\pnnp{\pap_x}}{\varphi} \in \fset_x \} \setminus \fset_{\pst_x}$ \;
    \lIf{$\fset_\alt = \emptyset$}{$\sts_x :=$ \detstsd($x$)}
    \lElse{$sts_x := \closed(\fset_\alt)$} \;
  }
\end{procedure}

\noindent\prcdr{}~\detsts{}($x$) determines the current status of a node~$x$.
Its result will always be~$\closed(\cdot)$ or~$\open(\cdot,\cdot)$.
If~$x$ is an $\alpha$/$\beta$-node or a state,
the procedure just calls the corresponding sub-procedure.
If~$x$ is a special node,
we determine the set~$\fset_\alt$ of all formulae~$\varphi$
such that~$\paa{\pnnp{\pap_x}}{\varphi}$ is in~$\fset_x$
but~$\varphi$ is not in the set of the parent state of~$x$.
If there is no such formula, that is~$\fset_\alt$ is the empty set,
we say that~$x$ is \emph{compatible} with its parent state~$\pst_x$.
Note that incompatibilities can only arise because of converse programs.

If~$x$ is compatible with~$\pst_x$, all is well,
so we determine its status via the corresponding sub-procedure.
Else we cannot connect~$\pst_x$ to a state
with~$\fset_x$ assigned to it in the putative model as explained in the overview,
and, thus, we can close~$x$.
That does not, however, mean that~$\pst_x$ is unsatisfiable;
maybe it is just missing some formulae.
We cannot extend~$\pst_x$ directly
as this may have side-effects elsewhere;
but to tell~$\pst_x$ what went wrong, we remember~$\fset_\alt$.
The meaning is that if we create an alternative node for~$\pst_x$
by adding the formulae in~$\fset_\alt$,
we might be more successful in building an interpretation.

\begin{procedure}[t]
  \caption{\detstsb{}($x$) for determining the status of an $\alpha$- or a $\beta$-node}
  \dontprintsemicolon
  \SetVline
  \Input{an $\alpha$- or a $\beta$-node~$x \in V$ with~$\unexp \not= \sts_x \not= \closed(\cdot)$}
  \Output{the new status of~$x$}
  \BlankLine
  let~$y_1, \dotsc, y_k \in V$ be all children of~$x$ \;
  $\alt := \bigcup_{i=1}^k \alt_{y_i}$ \;
  \lIf{$\forall i \in \{ 1, \dotsc, k \}.\: \sts_{y_i} = \closed(\cdot)$}{\Return $\closed(\alt)$} \;
  \Else(\bcmt{} at least one child is not closed \ecmt{}){
    $\prs := \eprs$ \;
    \ForEach{$\varphi \in \fset_x \cap \fmlev$}{
      \lFor{$i \longleftarrow 1$ \KwTo $k$}{$\prsval_i :=$ \detstscint($x, y_i, \varphi$)} \;
      $\prsval :=$ \lIf{$\exists i \in \{ 1, \dotsc, k \}.\: \prsval_i = \bot$}{$\bot$} \lElse{$\bigcup_{i=1}^k \prsval_i$} \;
      $\prs := \updfkt{\prs}{\varphi}{\prsval}$ \;
    }
    $\prs' :=$ \filter($x, \prs$) \;
    \Return $\open(\prs', \alt)$ \;  
  }
\end{procedure}

\noindent\prcdr{}~\detstsb{}($x$)
computes the status of an $\alpha$- or a $\beta$-node~$x \in V$.
For this task, an $\alpha$-node can be seen as a $\beta$-node with exactly one child.
The set of alternative sets of~$x$
is the union of the sets of alternative sets of all children.
If all children of~$x$ are closed then~$x$ must also be closed.
Otherwise we compute the set of potential rescuers
for each eventuality~$\varphi$ in~$\fset_x$ as follows.
For each child~$y_i$ of~$x$
we determine the potential rescuers of~$\varphi$ which result from following~$y_i$
by invoking \detstscint{}.
If the set of potential rescuers corresponding to some~$y_i$ is~$\bot$
then~$\varphi$ can currently be fulfilled via~$y_i$ and~$\prs_x(\varphi)$ is set to~$\bot$.
Else~$\varphi$ cannot currently be fulfilled in~$G$,
but each child returned a set of potential rescuers,
and the set of potential rescuers for~$\varphi$ is their union.
Finally, we deal with potential rescuers in~$\prs$
of the form~$(x, \chi)$ for some~$\chi \in \fmlev$ by calling \filter{}.

\begin{procedure}[t]
  \caption{\detstss{}($x$) for determining the status of a state}
  \dontprintsemicolon
  \SetVline
  \Input{a state~$x \in V$ with~$\unexp \not= \sts_x \not= \closed(\cdot)$}
  \Output{the new status of~$x$}
  \BlankLine
  let~$\pea{l_1}{\varphi_1}, \dotsc , \pea{l_k}{\varphi_k}$ be all of the \feal-formulae in~$\fset_x$ \;
  \lFor{$i \longleftarrow 1$ \KwTo $k$}{$y_i := \chld(x, \pea{l_i}{\varphi_i})$} \;
  \lIf{$\exists i \in \{ 1, \dotsc, k \}.\: \sts_{y_i} = \closed(\alt)$}{\Return $\closed(\alt)$} \;
  \Else(\bcmt{} no child is closed \ecmt{}){
    $\alt := \bigcup_{i=1}^k \alt_{y_i}$ \;
    $\prs := \eprs$ \;
    \For{$i \longleftarrow 1$ \KwTo $k$}{
      \If{$\varphi_i \in \fmlev$}{
        $\prsval :=$ \detstscint($x, y_i, \varphi_i$) \;
        $\prs := \updfkt{\prs}{\pea{l_i}{\varphi_i}}{\prsval}$ \;
      }
    }
    $\prs' :=$ \filter($x, \prs$) \;
    \Return $\open(\prs', \alt)$ \;  
  }
\end{procedure}

\noindent\prcdr{}~\detstss{}($x$) computes the status of a state~$x \in V$.
We obtain the successors for all \feal-formulae in~$\fset_x$.
If any successor is closed then~$x$ is closed with the same set of alternative sets.
Else the set of alternative sets of~$x$
is the union of the sets of alternative sets of all children
and we compute the potential rescuers
for each eventuality~$\pea{l_i}{\varphi_i}$ in~$\fset_x$
by invoking \detstscint{}.
Finally, we deal with potential rescuers in~$\prs$
of the form~$(x, \chi)$ for some~$\chi \in \fmlev$ by calling \filter{}.
Note that we do not consider eventualities which are not \feal-formulae.
The intuitive reason is that the potential rescuers of such eventualities
are determined by following the annotation chain (see below).
However, different special nodes which have the same set, and hence all link to~$x$,
might have different annotations.
Hence we cannot (and do not need to) fix the potential rescuer sets
for eventualities in~$x$ which are not \feal-formulae.

\begin{procedure}[t]
  \caption{\detstsd($x$) for determining the status of a special node}
  \dontprintsemicolon
  \SetVline
  \Input{a special node~$x \in V$ with~$\unexp \not= \sts_x \not= \closed(\cdot)$}
  \Output{the new status of~$x$}
  \BlankLine
  $y_0 := \chld(x, \mrk)$ \;
  let~$\fset_1, \dotsc , \fset_{j}$ be all the sets in the set~$\alt_{y_0}$ \;
  \For{$i \longleftarrow 1$ \KwTo $j$}{
    $y_i := \chld(x, \fset_i)$ \;
    \If(\bcmt{} child does not exist \ecmt{}){$y_i = \bot$}{
      $y_i := \create(\fset_x \cup \fset_i, \ann_x, \pst_x, \pap_x, \bot, \unexp)$ \;
      insert~$y_i$, and an edge from~$x$ to~$y_i$ labelled with~$\fset_i$, into~$G$ \;
    }
  }
  let~$y_{j+1}, \dotsc, y_k$ be all the remaining children of~$x$ \;
  $\alt := \bigcup_{i=1}^k \alt_{y_i}$ \;
  \lIf{$\forall i \in \{ 0, \dotsc, k \}.\: \sts_{y_i} = \closed(\cdot)$}{\Return $\closed(\alt)$} \;
  \Else(\bcmt{} at least one child is not closed \ecmt{}){
    $\prs := \eprs$ \;
    \ForEach{$\varphi \in \fset_x \cap \fmlev$}{
      $\varphi' := \annchn(x, \varphi)$ \;
      \If{$\varphi' \in \fmlev$}{
        \lFor{$i \longleftarrow 0$ \KwTo $k$}{$\prsval_i :=$ \detstscint($x, y_i, \varphi'$)} \;
        $\prsval :=$ \lIf{$\exists i \in \{ 0, \dotsc, k \}.\: \prsval_i = \bot$}{$\bot$} \lElse{$\bigcup_{i=0}^k \prsval_i$} \;
        $\prs := \updfkt{\prs}{\varphi}{\prsval}$ \;
      }
    }
    $\prs' :=$ \filter($x, \prs$) \;
    \Return $\open(\prs', \alt)$ \;  
  }
\end{procedure}

\noindent\prcdr{}~\detstsd($x$)
computes the status of a special node~$x \in V$.
First, we retrieve the state~$y_0$ corresponding to~$x$,
namely the unique $\mrk$-child of~$x$.
For all alternative sets~$\fset_i$ of~$y_0$ we do the following.
If there does not exist a child of~$x$
such that the corresponding edge is labelled with~$\fset_i$,
we create a new node~$y_i$
whose associated set is the result of adding the formulae in~$\fset_i$ to~$\fset_x$.
The annotations, the parent state, and the parent program of~$y_i$ are inherited from~$x$.
We label the new edge from~$x$ to~$y_i$ with~$\fset_i$.
In other words we unpack the information
stored in the alternative sets in~$\alt_{y_0}$ into actual nodes
which are all children of~$x$.
Note that each~$\fset_i \neq \emptyset$ by construction in \detsts{}.
Some children of~$x$ may not be referenced from~$\alt_{y_0}$,
but we consider them anyway.

\begin{procedure}[p]
  \caption{\detstscint{}($x, y, \varphi$) for passing a $\prs$-entry of a child to a parent}
  \dontprintsemicolon
  \SetVline
  \Input{two nodes~$x, y \in V$ and a formula~$\varphi \in \fset_y \cap \fmlev$}
  \Output{$\bot$ or a set of node-formula pairs}
  \Remark{if \detstscint{}($x, y, \varphi$) has been invoked before
    with exactly the same arguments and
    \emph{under the same invocation of \mbox{\rm \detstsb{}}, \mbox{\rm \detstss{}}
    or \mbox{\rm \detstsd{}}},
    the procedure is not executed a second time
    but returns the cached result of the first invocation.
    We do not model this behaviour explicitly in the pseudocode.}
  \BlankLine
  \lIf{$\sts_y = \closed(\cdot)$}{\Return $\emptyset$} \;
  \lElseIf{$\sts_y = \unexp \xor \sts_y = \undef \xor \xnot y \nords x$}{\Return $\{ (y, \varphi) \}$} \;
  \Else(\bcmt{} \mbox{$\sts_y = \open(\cdot,\cdot) \xand y \nords x$} \ecmt{}){
    \lIf{$\prs_y(\varphi) = \bot$}{\Return $\bot$} \;
    \Else(\bcmt{} \mbox{$\prs_y(\varphi)$} is defined \ecmt{}){
      let~$(z_1, \varphi_1), \dotsc, (z_k, \varphi_k)$ be all of the pairs in~$\prs_y(\varphi)$\;
      \lFor{$i \longleftarrow 1$ \KwTo $k$}{$\prsval_i :=$ \detstscint($x, z_i, \varphi_i$)} \;
      \lIf{$\exists j \in \{ 1, \dotsc, k \}.\: \prsval_j = \bot$}{\Return $\bot$} \lElse{\Return $\bigcup_{i=1}^k \prsval_i$} \;
    }
  }
\end{procedure}

The set of alternative sets of~$x$
is the union of the sets of alternative sets of all children;
with the exception of~$y_0$
since the alternative sets of~$y_0$ are not related to~$\pst_x$
but affect~$x$ directly as we have seen.
If all children of~$x$ are closed then~$x$ must also be closed.
Otherwise we compute the set of potential rescuers
for each eventuality~$\varphi$ in~$\fset_x$ as follows.

First, we determine~$\varphi' := \annchn(x, \varphi)$.
Note that~$\varphi'$ is defined
because the special node~$x$ cannot contain an ``at a world'' cycle by definition.
If~$\varphi'$ is not an eventuality
then~$\varphi'$ is fulfilled in~$x$ and~$\prs(\varphi)$ remains~$\bot$.
If~$\varphi'$ is an eventuality, it must be a \feal-formula as~$x$ is a special node.
We use~$\varphi'$ instead of~$\varphi$
since only \feal-formula have a meaningful interpretation in~$\prs_{y_0}$ (see above).
For each child~$y_i$ of~$x$
we determine the potential rescuers of~$\varphi'$ by invoking \detstscint{}.
If the set of potential rescuers corresponding to some~$y_i$ is~$\bot$
then~$\varphi'$ can currently be fulfilled via~$y_i$ and so~$\prs_x(\varphi)$ is set to~$\bot$.
Otherwise~$\varphi'$ cannot currently be fulfilled in~$G$,
but each child returned a set of potential rescuers,
and the set of potential rescuers for~$\varphi$ is their union.
Finally, we deal with potential rescuers in~$\prs$
of the form~$(x, \chi)$ for some~$\chi \in \fmlev$ by calling \filter{}.

\noindent\prcdr{}~\detstscint{}($x, y, \varphi$)
determines whether an eventuality~$\psi \in \fset_x$,
which is not passed as an argument,
can be fulfilled via~$y$
such that~$\varphi$ is part of the corresponding fulfilling path;
or else which potential rescuers~$\psi$ can reach via~$y$ and~$\varphi$.
If~$y$ is closed, it cannot help to fulfil~$\psi$
as indicated by the empty set.
If~$x$ is undefined or did not become defined before~$x$
then~$(y, \varphi)$ itself is a potential rescuer of~$x$.
Else, if~$\varphi$ can be fulfilled, i.e.~$\prs_y(\varphi) = \bot$,
then~$\psi$ can be fulfilled too, so we return~$\bot$.
Otherwise we invoke the procedure recursively
on all potential rescuers in~$\prs_y(\varphi)$.
If at least one of these invocations returns~$\bot$
then~$\psi$ can be fulfilled via~$y$ and~$\varphi$
and the corresponding rescuer in~$\prs_y(\varphi)$.
If all invocations return a set of potential rescuers,
the set of potential rescuers for~$\psi$ is their union.
The recursion is well-defined
because if~$(z_i, \varphi_i) \in \prs_y(\varphi)$
then either~$z_i$ is still undefined or~$z_i$ became defined later than~$y$.

Each invocation of \detstscint{} can be uniquely assigned
to the invocation of \detstsb{}, \detstss{}, or \detstsd{}
which (possibly indirectly) invoked it.
To meet our complexity bound,
we require that under the same invocation of \detstsb{}, \detstss{}, or \detstsd{},
the procedure \detstscint{} is only executed at most once for each argument triple.
Instead of executing it a second time with the same arguments,
it uses the cached result of the first invocation.
Since \detstscint{} does not modify the graph,
the second invocation would return the same result as the first one.
An easy implementation of the cache is
to store the result of \detstscint{}($x, y, \varphi$) in the node~$y$
together with~$\varphi$ and a unique id number
for each invocation of \detstsb{}, \detstss{}, or \detstsd{}.

\begin{procedure}[t]
  \caption{\filter{}($x, \prs$) for handling self-loops in~$\prs$ chains in~$G$}
  \dontprintsemicolon
  \SetVline
  \Input{a node~$x \in V$ and a function~$\prs: \fmlev \to \opt{(\pows{V \times \fmlev})}$}
  \Output{$\prs$ where self-loops have been handled}
  \BlankLine
  $\prs' := \eprs$ \;
  \ForEach{$\varphi \in \fset_x \cap \fmlev$ such that $\prs(\varphi) \not= \bot$}{
    $\Delta := \{ \varphi \} \cup \prsreach(\prs, x, \varphi)$ \;
    \If{$\xnot \exists \chi \in \Delta.\: \prs(\chi) = \bot$}{
      $\prsval := \bigcup_{\chi \in \Delta} \big\{ (z, \psi) \in \prs(\chi) \mid z \not= x \big\}$\;
      $\prs' := \updfkt{\prs'}{\varphi}{\prsval}$ \;
    }
  }
  \Return $\prs'$ \;
\end{procedure}

\noindent\prcdr{}~\filter{}($x, \prs$) deals with the potential rescuers
for each eventuality of a node~$x$
which are of the form~$(x, \psi)$ for some~$\psi \in \fmlev$.
The second argument of \filter{} is a provisional~$\prs$ for~$x$.
If an eventuality~$\varphi \in \fset_x$ is currently fulfillable in~$G$
there is nothing to be done,
so let~$(x, \psi) \in \prs(\varphi)$.
If~$\psi = \varphi$ then~$(x, \varphi)$ cannot be a potential rescuer for~$\varphi$ in~$x$
and should not appear in~$\prs(\varphi)$.
But what about potential rescuers of the form~$(x, \psi)$ with~$\psi \not= \varphi$?
Since we want the nodes in the potential rescuers
to become defined later than~$x$,
we cannot keep~$(x, \psi)$ in~$\prs(\varphi)$;
but we cannot just ignore the pair either.

Intuitively~$(x, \psi) \in \prs(\varphi)$ means
that~$\varphi \in \fset_x$ can ``reach''~$\psi \in \fset_x$ by following a loop in~$G$
which starts at~$x$ and returns to~$x$ itself.
Thus if~$\psi$ can be fulfilled in~$G$, so can~$\varphi$;
and all potential rescuers of~$\psi$ are also potential rescuers of~$\varphi$.
The function~$\prsreach(\prs, x, \varphi)$ computes all eventualities in~$x$
which are ``reachable'' from~$\varphi$ in the sense above,
where transitivity is taken into account.
That is, it detects all self-loops from~$x$ to itself
which are relevant for fulfilling~$\varphi$.
We add~$\varphi$ as it is not in~$\prsreach(\prs, x, \varphi)$.
If any of these eventualities is fulfilled in~$G$
then~$\varphi$ can be fulfilled and is consequently undefined in the resulting~$\prs'$.
Otherwise we take all their potential rescuers
whose nodes are not~$x$.

 \begin{theorem}[Soundness, Completeness and Complexity]
  Let~$\phi \in \fml$ be a formula in negation normal form of size~$n$.
  The procedure $\issat(\phi)$ terminates, runs in \exptime{} in~$n$,
  and~$\phi$ is satisfiable iff $\issat(\phi)$ returns true.
\end{theorem}

\section{Implementation, Optimisations, and Strategy}
\label{secd:strategy}

It should be fairly straightforward to implement our algorithm.
It remains to show
an  efficient way to find nodes which are not up-to-date.
It is not too hard to see 
that the status of a node~$x$ can become outdated
only if its children change their status or 
\detstscint($x, y, \cdot$) was invoked when~$x$'s previous status was determined
and~$y$ now changes its status.
If we keep track of nodes of the second kind by inserting additional ``update''-edges
as described in \cite{gore-widmann-pdl-cade09},
we can use a queue for all nodes that might need updating.
When the status of a node is modified,
we queue all parents and all nodes linked by ``update''-edges.

We have omitted several refinements from our description for clarity.
The most important is that if a state~$s$ is closed,
all non-states which have~$s$ as a parent state are ignorable
since their status cannot influence any other node~$t$
unless~$t$ also has~$s$ as a parent state.
Moreover, if every special node parent~$x$ of a state~$s'$ is
incompatible or itself has a closed parent state,
then~$s'$ and the nodes having~$s'$ as parent state are ignorable.
This applies transitively,
but if~$s'$ gets a new parent whose parent state is not closed
then~$s'$ becomes ``active'' again.



Another issue is which rule to choose if several are applicable.
As we have seen, it is advantageous to close nodes as early as possible.
Apart from immediate contradictions,
we have Rule~4 which closes a node because it contains an unfulfillable eventuality.
If we can apply Rule~4 early while the graph is still small,
we might prevent big parts of the graph being built needlessly later.
Trying to apply Rule~4 has several consequences on the strategy of how to apply rules.

First, it is important to keep all nodes up-to-date
since Rule~4 is not applicable otherwise.
Second, it is preferable that a node~$x$ cannot reach open nodes
which became defined (or will be defined) after~$x$ did.
Hence, we should try to use Rule~2 on a node only if all children are already defined.

\section{An Example}

To demonstrate how the algorithm works,
we invoke it on the satisfiable toy formula~$\pea{a}{\phi}$ where~$\phi := \pea{\prp{a}}{\paa{\pcv{a}}{p}}$.
To save space, Fig.~\ref{figd:ex1} only shows the core subgraph of the tableau.
Remember that the order of rule applications is not fixed
but the example will follow some of the guidelines given in Section~\ref{secd:strategy}.
\begin{figure}[p]
  \centering
  \begin{math}
    \xymatrix{
      *++[F]{
        \begin{tabular}{c}
          (1) \hfill\phantom{M} state \\
          $\{\: \phi, \ \pea{a}{\phi} \:\}$ \\
          $\bot, \bot$ \hfill $5$
        \end{tabular}
      }
      \ar[r]^-{\pea{a}{\phi}}
      &
      *++[F]{
        \begin{tabular}{c}
          (2) \hfill\phantom{M} $\beta$-node \\
          $\{\: \phi \:\}$ \\
          $1, a$ \hfill $4$
        \end{tabular}
      }
      \ar[r]
      \ar[d]
      &
      *++[F]{
        \begin{tabular}{c}
          (3) \hfill\phantom{I} special node \\
          $\{\: \phi \pzz \paa{\pcv{a}}{p} \:\}$ \\
          $1, a$ \hfill $2$
        \end{tabular}
      }
      \ar[d]_-{\mrk}
      \\
      *++[F]{
        \begin{tabular}{c}
          (6) \hfill\phantom{I} special node \\
          $\{\: \phi \pzz \pea{a}{\phi}, \ p \:\}$ \\
          $1, a$ \hfill $9$
        \end{tabular}
      }
      \ar[d]_-{\mrk}
      &
      *++[F]{
        \begin{tabular}{c}
          (4) \hfill\phantom{I} special node \\
          $\{\: \phi \pzz \pea{a}{\phi} \:\}$ \\
          $1, a$ \hfill $3$
        \end{tabular}
      }
      \ar[lu]_-{\mrk}
      \ar[l]_-{\{p\}} 
      &
      *++[F]{
        \begin{tabular}{c}
          (5) \hfill\phantom{M} state \\
          $\{\: \phi, \ \paa{\pcv{a}}{p} \:\}$ \\
          $\bot, \bot$ \hfill $1$
        \end{tabular}
      }
      \\
      *++[F]{
        \begin{tabular}{c}
          (7) \hfill\phantom{M} state \\
          $\{\: \phi, \ \pea{a}{\phi}, \ p \:\}$ \\
          $\bot, \bot$ \hfill $8$
        \end{tabular}
      }
      \ar[r]^-{\pea{a}{\phi}}
      &
      *++[F]{
        \begin{tabular}{c}
          (8) \hfill\phantom{M} $\beta$-node \\
          $\{\: \phi \:\}$ \\
          $7, a$ \hfill $7$
        \end{tabular}
      }
      \ar[r]
      \ar[ld]
      &
      *++[F]{
        \begin{tabular}{c}
          (9) \hfill\phantom{I} special node \\
          $\{\: \phi \pzz \paa{\pcv{a}}{p} \:\}$ \\
          $7, a$ \hfill $6$
        \end{tabular}
      }
      \ar[u]^-{\mrk}
      \\
      *++[F]{
        \begin{tabular}{c}
          (10) \hfill\phantom{I} special node \\
          $\{\: \phi \pzz \pea{a}{\phi} \:\}$ \\
          $7, a$ \hfill $10$
        \end{tabular}
      }
      \ar`l/20pt[uuu] `[uuu]_-{\mrk} [uuu]
      \ar[r]^-{\{p\}} 
      &
      *++[F]{
        \begin{tabular}{c}
          (11) \hfill\phantom{I} special node \\
          $\{\: \phi \pzz \pea{a}{\phi}, \ p \:\}$ \\
          $7, a$ \hfill $11$
        \end{tabular}
      }
      \ar`u/10pt[l] `[ul]_<<{\mrk} [lu]
    }
  \end{math}
  \caption[]{An example: The graph~$G$ just before setting the status of node~(2)}
  \label{figd:ex1}
\end{figure}
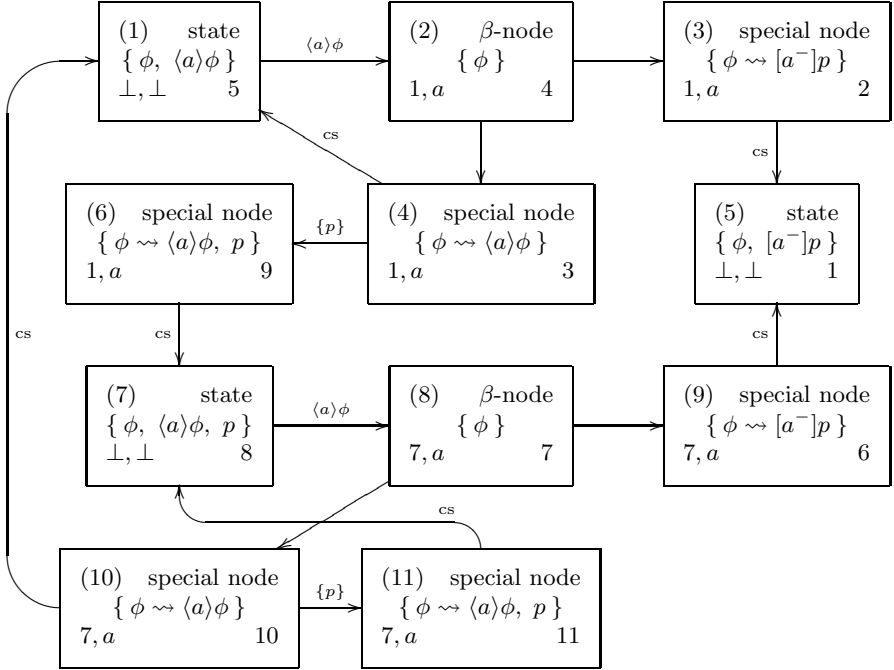

The nodes in Fig.~\ref{figd:ex1} are numbered in order of creation.
The annotation~$\ann$ is given using~``$\pzz$'' in~$\fset$.
For example, in node~(3), we have~$\fset_3 = \{\: \phi,\ \paa{\pcv{a}}{p} \:\}$,
and~$\ann_3$ maps the eventuality~$\phi$ to~$\paa{\pcv{a}}{p}$
and is undefined elsewhere.
The bottom line of a node contains the parent state and the parent program on the left,
and the time stamp on the right.
We do not show the status of a node 
since it changes during the algorithm, but explain it in the text.
If we write~$\sts_x = \open(\Lambda,\cdot)$ where~$\Lambda \subseteq V \times \fmlev$,
we mean that~$\prs_x$ maps all eventualities in~$\fset_x$,
with the exception of non-\feal-formulae if~$x$ is a state, to~$\Lambda$
and is undefined elsewhere.

We only consider the core subgraph of~$\phi$ and start by expanding node~(1)
which creates~(2).
Then we expand~(2) and create~(3) and~(4) which are both special nodes.
Next we expand~(3) and create the state~(5).
Expanding~(5) creates no new nodes since~$\fset_5$ contains no \feal-formula.
Now we define~(5) and then~(3).
This results in setting~$\sts_5 := \open(\eprs, \emptyset)$ according to \detstss{},
and~$\sts_3 := \closed(\{p\})$ since~(3) is not compatible with its parent state~(1).
Expanding~(4) inserts the edge from~(4) to~(1)
and defining~(4) sets~$\sts_4 := \open(\{(1, \pea{a}{\phi})\}, \emptyset)$ according to \detstsd{}.
Note that~(6) does not exist yet.
Next we define~(2) and then~(1)
which results in setting~$\sts_2 := \open(\{(1, \pea{a}{\phi})\}, \{p\})$ according to \detstsb{}
and~$\sts_1 := \open(\emptyset, \{p\})$ thanks to \filter{}.

Note that~$\pea{a}{\phi} \in \fset_1$ has an empty set of potential rescuers.
In PDL, we could thus close~(1),
but converse programs complicate matters for CPDL
as reflected by the fact that Rule~4 is not applicable for~(1)
because~(4) is not up-to-date.
Updating~(4) creates~(6) and sets~$\sts_4 := \open(\{(1, \pea{a}{\phi}), (6, \pea{a}{\phi})\}, \emptyset)$.
Updating~(2) and then~(1) sets~$\sts_2 := \open(\{(1, \pea{a}{\phi}), (6, \pea{a}{\phi})\}, \{p\})$
and~$\sts_1 := \open(\{(6, \pea{a}{\phi})\}, \{p\})$.
Now all nodes are up-to-date,
but Rule~4 is not applicable for~(1)
because the set of potential rescuers for~$\phi$ is no longer empty.

Next we expand~(6), which creates~(7),
then~(7), which creates~(8),
then~(8), which creates~(9) and~(10),
and finally~(9), which creates no new nodes.
Node~(9) is similar to~(3), but unlike~(3),
it is compatible with its parent state~(7)
which results in~$\sts_9 := \open(\bot, \emptyset)$.
Using our strategy from the last section,
we would now expand~(10)
so that~(8) can become defined after both its children became defined.
Since~(9) fulfils all its eventualities,
we choose to define~(8) instead and set~$\sts_8 := \open(\bot, \emptyset)$.
Next we define~(7) and then~(6)
which sets~$\sts_7 := \open(\bot, \emptyset)$ and~$\sts_6 := \open(\bot, \emptyset)$.
The status of~(4) is not affected since~(6) was defined after~(4),
giving ``(6) $\not\nords$ (4)'' in \detstscint(4, 6, $\pea{a}{\phi}$).

We expand~(10) which inserts the edge from~(10) to~(1).
Then we define~(10) which creates~(11)
and sets~$\sts_{10} := \open(\bot, \emptyset)$.
Note that the invocation of \detstscint(10, 1, $\pea{a}{\phi}$)
in the invocation \detstsd(10)
leads to the recursive invocation \detstscint(10, 6, $\pea{a}{\phi}$).
Expanding and defining~(11) yields~$\sts_{11} := \open(\bot, \emptyset)$.
Finally, no rule is applicable in the shown subgraph.

\bibliographystyle{splncs}
\bibliography{cpdl}

\clearpage

\appendix

\section{Soundness, Completeness, and Complexity}

We start by some more definitions.

\begin{definition}
  We define the size~$|\cdot|$ on formulae and programs
  in negation normal form inductively as follows:
  \begin{itemize}
  \item $|p| := |\pnot{p}| := |a| := |\pcv{a}| := 1$
  \item $|\pand{\varphi}{\psi}| := |\por{\varphi}{\psi}| := 1 + |\varphi| + |\psi|$
  \item $|\pea{\gamma}{\varphi}| := |\paa{\gamma}{\varphi}| := |\gamma| + |\varphi|$
  \item $|\psp{\gamma}{\delta}| := |\pup{\gamma}{\delta}| := 1 + |\gamma| + |\delta|$
  \item $|\pip{\varphi}| := 1 + |\varphi|$
  \item $|\prp{\gamma}| := 1 + |\gamma| \enspace.$
  \end{itemize}
\end{definition}

\begin{definition}
  Let~$\varphi \in \fml$ be a formula in negation normal form.
  The \emph{closure}~$\pcl(\varphi)$ is the least set of formulae such that:
  \begin{itemize}
  \item $\varphi \in \pcl(\phi)$
  \item $\pea{l}{\psi} \in \pcl(\varphi) \xor \paa{l}{\psi} \in \pcl(\varphi) \ximp \psi \in \pcl(\varphi)$
  \item $\alpha \in \pcl(\varphi) \ximp \alpha_1 \in \pcl(\varphi) \xand \alpha_2 \in \pcl(\varphi)$
  \item $\beta \in \pcl(\varphi) \ximp \beta_1 \in \pcl(\varphi) \xand \beta_2 \in \pcl(\varphi) \enspace.$
  \end{itemize}
\end{definition}
It is easy to see that all formula in~$\pcl(\varphi)$ are also in negation normal form.

\begin{definition}
  A \emph{transition frame} is a pair~$(W, R)$
  where~$W$ is a non-empty set of worlds and~$R: \aprg \to W \times W$ is a function
  mapping each atomic program~$a \in \aprg$ to a binary relation~$R_a$ over~$W$.
  We extend~$R$ to~$\lprg$ by defining~$R_{\pcv{a}} := \{ (w, v) \mid (v, w) \in R_a \}$.
  A \emph{model}~$(W, R, \evv)$ is a transition frame~$(W, R)$
  and a valuation function~$\evv: \afml \to 2^W$
  mapping each propositional variable~$p \in \afml$ to a set~$\evv(p)$ of worlds
  where~$p$ is ``true''.
\end{definition}

\begin{definition}
  A formula~$\varphi \in \fml$ is \emph{satisfiable}
  iff there exists a model~$M = (W, R, \evv)$ and a world~$w \in W$ s.t.~$\psr{M, w}{\varphi}$,
  and \emph{valid} iff~$\pnot{\varphi}$ is unsatisfiable.
\end{definition}

\begin{proposition}
 \label{propd:axioms}
  In the notation of Table~\ref{tabd:alphabeta},
  the formulae of the form~$\peq{\alpha}{\pand{\alpha_1}{\alpha_2}}$ and~$\peq{\beta}{\por{\beta_1}{\beta_2}}$ are valid.
\end{proposition}

\begin{definition}
  \label{defd:event}
  For a given~$\varphi \in \fml$, the (infinite) set~$\ppre(\varphi)$ is defined as:
  \[ \ppre(\varphi) := \{ \psi \in \fml \mid \exists k \in \Nat.\: \exists \gamma_1, \dotsc, \gamma_k \in \prg.\: \psi = \pea{\gamma_1}{\dotsc \pea{\gamma_k}{\varphi}} \} \enspace.\]
  The set of all \emph{eventualities} is defined as:
  \[ \fmlev := \bigcup_{\varphi \in \Delta}\ppre(\varphi) \text{ where } \Delta := \{ \pea{\prp{\gamma}}{\psi} \mid \gamma \in \prg \xand \psi \in \fml \} \enspace. \]
  For all~$\varphi, \psi \in \fml$, the binary relation~$\pzz$ on formulae is defined as:
  $\varphi \pzz \psi$ iff (exactly) one of the following conditions is true:
  \begin{itemize}
  \item $\exists \chi \in \fml.\: \exists \gamma, \delta \in \prg.\: \varphi = \pea{\psp{\gamma}{\delta}}{\chi} \xand \psi = \pea{\gamma}{\pea{\delta}{\chi}}$
  \item $\exists \chi \in \fml.\: \exists \gamma, \delta \in \prg.\: \varphi = \pea{\pup{\gamma}{\delta}}{\chi} \xand
    \big( \psi = \pea{\gamma}{\chi} \xor \psi = \pea{\delta}{\chi} \big)$
  \item $\exists \chi \in \fml.\: \exists \gamma \in \prg.\: \varphi = \pea{\prp{\gamma}}{\chi} \xand
    \big( \psi = \chi \xor \psi = \pea{\gamma}{\pea{\prp{\gamma}}{\chi}} \big)$
  \item $\exists \chi, \phi \in \fml.\: \varphi = \pea{\pip{\phi}}{\chi} \xand \psi = \chi \enspace.$
  \end{itemize}
\end{definition}

Intuitively, using Table~\ref{tabd:alphabeta},
the~``$\pzz$'' relates a $\pea{}{}$-formulae~$\alpha$ (respectively~$\beta$),
to~$\alpha_1$ (respectively~$\beta_1$ and~$\beta_2$)
while~$\ppre(\varphi)$ captures
that~$\pea{\prp{\gamma}}{\varphi}$ can be ``reduced'' to~$\pea{\gamma}{\pea{\prp{\gamma}}{\varphi}}$,
which can be reduced to~$\pea{\gamma_1}{\dotsc \pea{\gamma_k}{\pea{\prp{\gamma}}{\varphi}}}$.
Note that~$\varphi \in \ppre(\varphi)$.

\begin{definition}
  A \emph{structure}~$(W, R, L)$ $[$for~$\varphi \in \fml]$ is a
  transition frame~$(W, R)$ and a labelling function~$L: W \to 2^{\fml}$
  which maps each world~$w \in W$ to a set~$L(w)$ of formulae
  $[$and has~$\varphi \in L(v)$ for some world~$v \in W]$.
\end{definition}

\begin{definition}
  Let~$H = (W, R, L)$ be a structure, $w \in W$, and~$\varphi, \psi \in \fml$.
  A \emph{fulfilling chain} for~$(H, w, \varphi, \psi)$
  is a finite sequence~$(w_0, \psi_0), \dotsc, (w_n, \psi_n)$
  of world-formula pairs with~$n \geq 0$ such that:
  \begin{itemize}
  \item $w_0 = w$, $\psi_0 = \varphi$, $\psi_n = \psi$, and~$\psi_i \in L(w_i)$ for all~$0 \leq i \leq n$
  \item for all~$0 \leq i < n$:
    if~$\psi_i = \pea{l}{\chi}$ for some~$l \in \lprg$ and some~$\chi \in \fml$
    then~$\psi_{i+1} = \chi$ and~$\prel{w_i}{R_b}{w_{i+1}}$, else~$\psi_i \pzz \psi_{i+1}$ and~$w_i = w_{i+1}$.
  \end{itemize}
\end{definition}

\begin{definition}
  A \emph{Hintikka structure}~$H = (W, R, L)$ $[$for~$\varphi \in \fml]$
  is a structure $[$for~$\varphi]$
  which satisfies the following conditions for every~$w \in W$,
  where~$\alpha$ and~$\beta$ are formulae as defined in Table~\ref{tabd:alphabeta}:
  \begin{displaymath}
    \begin{array}{ll}
      \mathrm{H1:} & \pnot{p} \in L(w) \ximp p \not\in L(w)\\
      \mathrm{H2:} & \alpha \in L(w) \ximp \alpha_1 \in L(w) \xand \alpha_2 \in L(w)\\
      \mathrm{H3:} & \beta \in L(w) \ximp \beta_1 \in L(w) \xor \beta_2 \in L(w)\\
      \mathrm{H4:} & \pea{l}{\varphi} \in L(w) \ximp \exists v \in W.\: \prel{w}{R_b}{v} \xand \varphi \in L(v)\\
      \mathrm{H5:} & \paa{l}{\varphi} \in L(w) \ximp \forall v \in W.\: \prel{w}{R_b}{v} \ximp \varphi \in L(v)\\
      \mathrm{H6:} & \pea{\prp{\gamma}}{\varphi} \in L(w) \ximp
      \text{there exists a fulfilling chain for~$(H, w, \pea{\prp{\gamma}}{\varphi}, \varphi)$} \enspace.
    \end{array}
  \end{displaymath}
\end{definition}

H3 ``locally unwinds'' the fix-point semantics of~$\pea{\prp{\gamma}}{\varphi}$,
but does not guarantee a \emph{least} fix-point
which requires~$\varphi$ be true eventually.
H6 ``globally'' ensures all $\pea{\prp{}}{}$-formulae are fulfilled.
H2 captures the \emph{greatest} fix-point semantics of~$\paa{\prp{\gamma}}{\varphi}$.

\begin{proposition}
  \label{propd:satisfiable}
  A formula~$\varphi \in \fml$ in negation normal form is satisfiable
  iff there exists a Hintikka structure for~$\varphi$.
\end{proposition}

Proposition~\ref{propd:satisfiable} implies
that we can check whether a formula~$\varphi$ in negation normal form is satisfiable
by systematically trying to build a Hintikka structure for~$\varphi$.
Our tableau rules are designed specifically for this purpose.

\begin{definition}
  Let~$G = (V, E)$ be a directed graph and~$x, y \in V$ two of its nodes.
  Then~$y$ is a \emph{child} of~$x$ iff~$(x, y) \in E$.
  A \emph{path}~$\pi$ in~$G$
  is a finite or infinite sequence~$x_0, x_1, x_2, \dotsc$ of nodes in~$G$
  such that~$x_{i+1}$ is a child of~$x_i$ for all~$x_i$
  except the last node if~$\pi$ is finite.
  An \emph{$x$-path}~$\pi$ is a path in~$G$ that has~$x_0 = x$.
\end{definition}

We now list some facts about the algorithm
which are needed in the subsequent proofs.
They can be verified by careful inspection of the procedures.
\renewcommand{\theenumi}{(\roman{enumi})}
\renewcommand{\labelenumi}{\theenumi}
\begin{proposition}
  \label{propd:facts}
  Let~$x, y, z \in V$ be nodes and~$\varphi, \psi \in \fmlev$.
  \begin{enumerate}
  \item\label{enumd:facts:b} If~$x$ is open and~$\prs_x(\varphi)$ is defined
    then for each~$(y, \psi) \in \prs_x(\varphi)$, we have~$x \nords y$ and~$\psi \in \fset_y$.
    Moreover, if~$y$ is a state then~$\psi$ is a \feal-formula.
  \item\label{enumd:facts:f} Let $\prs' :=$ \mbox{\rm \filter($x, \prs$)}.
    Then~$\prs(\varphi) = \prs(\psi) \ximp \prs'(\varphi) = \prs'(\psi)$.
  \item\label{enumd:facts:d} The number of consecutive non-states in a path in~$G$ is bounded.
  \item\label{enumd:facts:e} if~$x$ and~$y$ are states with~$\fset_x = \fset_y$ then~$x = y$;
  \item\label{enumd:facts:c} If~$x$ is a state
    then its parents are exactly the special nodes~$y$ with~$\fset_y = \fset_x$.
  \item\label{enumd:facts:g} If~$x$ is a special node,
    it has the state~$y$ with~$\fset_y = \fset_x$ as child.
  \item\label{enumd:facts:a} If~$y$ is a child of~$x$ and neither of them are states
    then~$\pst_x = \pst_y$ and~$\pap_x = \pap_y$ and~$\fset_x \subsetneqq \fset_y$.
  \item\label{enumd:facts:h} We have~$\alt_\rt = \emptyset$ and~$\emptyset \notin \alt_x$.
  \end{enumerate}
\end{proposition}

Let~$\phi \in \fml$ be a formula in negation normal form.
Furthermore let~$G = (V, E)$ be the final graph, with root node~$\rt$,
which was created by invoking \issat($\phi$).
Note that~$\fset_\rt = \{ \pea{d}{\phi} \}$
and that all nodes in~$G$ are defined once our algorithm terminates.
Hence ``not closed'' then means ``open'' and vice versa,
and~$\nords$ is a total strict order.
Of course, we have to show first that our algorithm terminates,
but this is an immediate consequence from Theorem~\ref{theod:term}.

\begin{theorem}
  \label{theod:term}
  The algorithm runs in \exptime{} in~$n := |\phi|$.
\end{theorem}
\begin{proof}
It is easy to see
that every node in~$G$ can contain only formulae of the closure~$\pcl(\phi)$.
Furthermore it is known
that~$|\pcl(\phi)| \leq n$ and that~$|\varphi| \leq n^2$ for all~$\varphi \in \pcl(\phi)$.
Hence there are at most~$2^n$ different sets of formulae
that can be assigned to the nodes of~$G$.
As a state is uniquely identified by its assigned set of formulae
due to Prop.~\ref{propd:facts}\ref{enumd:facts:e},
the total number of states in~$G$ is also in~$2^{\bigo(n)}$.

If we fix one state~$x \in V$, 
it is not too hard to see that the nodes
which are not states and which have~$x$ as their parent state
form several disjoint trees
whose roots are exactly the children of~$x$.
Because of Prop.~\ref{propd:facts}\ref{enumd:facts:a},
the depth of these trees is in~$\bigo(n)$.
The branching degree of $\alpha$- and $\beta$-nodes is constant,
the branching degree of a special node, however, is in~$2^{\bigo(n)}$.
Hence the size of each tree is in~$\left( 2^{\bigo(n)} \right)^{\bigo(n)} \in 2^{\bigo(n^2)}$.
Since the number of children for each state is clearly in~$\bigo(n)$,
the total number of non-states 
which have~$x$ as their parent state is in~$\bigo(n) \cdot 2^{\bigo(n^2)} \in 2^{\bigo(n^2)}$.
As all non-states have a parent state,
we conclude that the total number of nodes in~$G$
is in~$2^{\bigo(n)} \cdot 2^{\bigo(n^2)} \in 2^{\bigo(n^2)}$.

Bearing that in mind, it is fairly obvious
that Rule~1, 2, and~4 can only be applied an exponential number of times.
Next we show that Rule~3 can only be applied an exponential number of times.
We do this by showing that each node can change from open to open
only an exponential number of times.

Assume for the sake of notation that the status of a node~$x$
changes from $\open(\alt, \prs)$ to~$\open(\alt', \prs')$.
In particular we have~$\alt \not= \alt'$ or~$\prs \not= \prs'$.

First we note that we have~$\alt \subseteq \alt'$
which basically follows from the fact
that the set of alternative sets of a node is always the union
of the sets of alternative sets of its children which are not states,
even when the nodes are closed.
The only exception are states
where the set of alternative sets can become smaller,
but only if the node becomes closed.
However, all parents of states are special nodes,
see Prop.~\ref{propd:facts}\ref{enumd:facts:e},
and special nodes do not inherit
the set of alternative of its corresponding state, see \detstsd{}.
As a consequent, $\alt_x$ can only change an exponential number of times
since there are only exponential many alternative sets.

Next we show that if there are~$\varphi, \psi \in \fmlev$ and~$y \in V$
such that either~$\prs(\varphi) = \bot \not= \prs'(\varphi)$
or~$(y, \psi) \in \prs(\varphi)$ but~$(y, \psi) \notin \prs'(\varphi) \not= \bot$,
then some node~$z \in V$ must have been closed
between the time when~$\prs$ was calculated
and the time when~$\prs'$ was calculated.
Hence, there can only be an exponential number of those changes of~$\prs_x$.
Changes of~$\prs_x$ which do not fall into the described category
``fill up'' the sets of potential rescuers.
Therefore, there can only be an exponential number of such changes of~$\prs_x$ in a row
before one of the described changes must happen.
It follows that~$\prs_x$ can only change an exponential number of times.

To show that some node was closed between the times
when~$\prs_x$ and~$\prs_x'$ were calculated, we use induction on~$\nords$.
That is, for each node~$y$ with~$y \nords x$,
we can assume the induction hypothesis.
When~$\prs_x$ was calculated via \detstsb{}, \detstss{}, or \detstsd{},
it can be shown -- basically because of the ``monotonicity'' of \filter{} --
that there must exist~$\varphi_i, \psi' \in \fmlev$ and~$y_i, y' \in V$
such that one of the following statements holds.
\begin{itemize}
\item $\detstscint(x, y_i, \varphi_i) = \bot$ at the time of calculating~$\prs_x$,
  but we have $\detstscint(x, y_i, \varphi_i) \not= \bot$ at the time of calculating~$\prs_x'$; or
\item $(y', \psi') \in \detstscint(x, y_i, \varphi_i)$ at the time of calculating~$\prs_x$,
  but we have $(y', \psi') \notin \detstscint(x, y_i, \varphi_i) \not= \bot$ at the time of calculating~$\prs_x'$.
\end{itemize}
We only cover the second case, the first one is similar.
Using the definition of \detstscint{},
the second case implies one of the following three cases:
\begin{description}
\item[Case~1] the node~$y_i$ was open at the time of calculating~$\prs_x$,
  but is closed at the time of calculating~$\prs_x'$;
\item[Case~2] $y_i \nords x$ and there exists a pair~$(z_j, \varphi_j)$
  such that~$(z_j, \varphi_j) \in \prs_{y_i}(\varphi_i)$ at the time of calculating~$\prs_x$,
  but~$(z_j, \varphi_j) \notin \prs_{y_i}(\varphi_i) \not= \bot$ at the time of calculating~$\prs_x'$;
\item[Case~3] $y_i \nords x$ and there exists a pair~$(z_j, \varphi_j)$
  such that we have $(y', \psi') \in \detstscint(x, z_j, \varphi_j)$ at the time of calculating~$\prs_x$,
  but $(y', \psi') \notin \detstscint(x, z_j, \varphi_j) \not= \bot$ at the time of calculating~$\prs_x'$.
\end{description}
In \emph{Case~1} we have found the node we are looking for, namely~$y_i$.
In \emph{Case~2} we can apply the induction hypothesis on~$y_i$
which gives us the desired node.
In \emph{Case~2} we can use an inductive argument 
which is well-defined because of Prop.~\ref{propd:facts}\ref{enumd:facts:b}.
We conclude that each rule of the algorithm can only be applied
an exponential number of times.

Next we show that applying a rule can be done in \exptime{}.
To do this it is obviously enough to show
that \detstsb{}, \detstss{}, and \detstsd{}
-- and hence \detsts{} -- run in \exptime{}.
It is not too hard to see that the runtime for \filter{} is in \exptime{}.
Hence we are left to show
that the direct invocations of \detstscint{}
in \detstsb{}, \detstss{}, or \detstsd{} run in \exptime{}.
We have already explained
that under each invocation of \detstsb{}, \detstss{}, and \detstsd{},
the procedure \detstscint{} is invoked
at most once for each pair~$(y, \varphi) \in V \times \fmlev$ as second and third argument.
As there exists only an exponential number of such pairs
and the runtime for \detstscint{} is clearly in \exptime{}
when ignoring recursive invocations,
the direct invocations of \detstscint{}
in \detstsb{}, \detstss{}, or \detstsd{} run in \exptime{}.

Since the number of nodes is exponential and \detsts{} runs in \exptime{},
checking whether one of the rules is applicable can clearly be done in \exptime{};
even with the most naive way of simply trying the conditions for all nodes.
We can therefore conclude that the algorithm runs in \exptime{}.
\qed
\end{proof}

\begin{lemma}
  \label{lemmad:uptodate}
  Let~$x \in V$ be an open node.
  \begin{enumerate}
  \item\label{enumd:uptodate:b} If~$x$ is a state then all of its children are open.
  \item\label{enumd:uptodate:c} If~$x$ is an $\alpha$- or a $\beta$-node
    or a special node then some child of~$x$ is open.
  \end{enumerate}
\end{lemma}
\begin{proof}
Since Rule~3 is not applicable, the node~$x$ is up-to-date.

\noindent\ref{enumd:uptodate:b}:
If some child of~$x$ were not open then it must be closed.
But then~$x$ would also be closed by definition of \detstss{}.
\noindent\ref{enumd:uptodate:c}:
If all children of~$x$ were not open then all children must be closed.
But then~$x$ would also be closed by definition of \detstsb{} or \detstsd{}.
\qed
\end{proof}

\begin{definition}
  Let~$x, y \in V$ and~$\varphi, \psi \in \fml$.
  A \emph{graph chain} for~$(x, \varphi, y, \psi)$
  is a finite sequence~$(y_0, \psi_0), \dotsc, (y_n, \psi_n)$
  of node-formula pairs with~$n \geq 0$ such that:
  \begin{itemize}
  \item $y_0 = x$, $\psi_0 = \varphi$, $y_n = y$, $\psi_n = \psi$,
    $y_i$ is open, and~$\psi_i \in \fset_{y_i}$ for all~$0 \leq i \leq n$
  \item $y_i = y_{i+1}$ or $y_{i+1}$ is a child of~$y_i$ for all~$0 \leq i < n$
  \item for all~$0 \leq i < n$, (exactly) one of the following conditions is true:
    \begin{itemize}
    \item $\psi_i = \psi_{i+1}$ and ($y_i = y_{i+1}$ or~$y_i$ is not a state);
    \item if~$\psi_i = \pea{l}{\chi}$ for some~$l \in \lprg$ and some~$\chi \in \fml$
      then~$\psi_{i+1} = \chi$ and~$y_i$ is a state and~$y_{i+1}$ is the successor of~$\pea{l}{\chi}$,
      else~$\psi_i \pzz \psi_{i+1}$ and~$y_i$ is not a state.
    \end{itemize}
  \end{itemize}
\end{definition}

\begin{lemma}
  \label{lemmad:ex-fulfilling}
  For every open node~$x \in V$ and every eventuality~$\varphi \in \fset_x \cap \fmlev$,
  where~$\varphi$ is a \feal-formula if~$x$ is a state, we have:
  \begin{enumerate}
  \item If~$\prs_x(\varphi) = \bot$ then there exists a node~$z \in V$,
    a formula~$\psi \in \fml \setminus \fmlev$ and a graph chain~$\fchn$ for~$(x, \varphi, z, \psi)$.
  \item If~$\prs_x(\varphi) \not= \bot$, we have for all~$(z, \psi) \in \prs_x(\varphi)$
    that there exists a graph chain~$\fchn$ for~$(x, \varphi, z, \psi)$.
  \end{enumerate}
\end{lemma}
\begin{proof}
We use induction on~$\nords$.
That is, for a node~$x \in V$ we can assume
that all nodes~$y \in V$ with~$y \nords x$ fulfil the lemma already.
Recall that~$y \nords x$ iff~$y$ becomes defined before~$x$ does.

Let~$x \in V$ be open and~$\varphi \in \fset_x$ be an eventuality.
We distinguish whether~$x$ is an $\alpha$/$\beta$-node or a state or a special node.

\noindent\emph{Case~1 ($x$ is $\alpha$/$\beta$-node):}
We distinguish whether~$\prs_x(\varphi) \not= \bot$ or~$\prs_x(\varphi) = \bot$.

\noindent\emph{Case~1.1 ($\prs_x(\varphi) \not= \bot$):}
Let~$(z, \psi) \in \prs_x(\varphi)$.
We distinguish whether or not there exists an open child~$y := y_i$ of~$x$
such that~$(z, \psi)$ is an element of the corresponding~$\prsval_{\varphi, i}$ in \detstsb($x$).

\noindent\emph{Case~1.1.1: ($y$ exists)}
In this case~$(z, \psi)$ must be in \detstscint($x, y, \varphi$).
Our final case distinction is whether or not~$y \nords x$.

\noindent\emph{Case~1.1.1.1: ($\xnot y \nords x$)}
In this case, we have~$y = z$ and~$\varphi = \psi$ by construction in \detstscint{}.
Hence~$(x, \varphi), (y, \varphi)$ is a graph chain for~$(x, \varphi, z, \psi)$.

\noindent\emph{Case~1.1.1.2: ($y \nords x$)}
In this case, we have~$\prs_y(\varphi) \not= \bot$ and
there must exist a pair~$(y', \varphi') \in \prs_y(\varphi)$
such that~$(z, \psi)$ is contained in \detstscint($x, y', \varphi'$).
Using the induction hypothesis on~$y$,
we obtain a graph chain~$\fchn_1$ for~$(y, \varphi, y', \varphi')$.
In particular~$y'$ is open
and we have~$y \nords y'$ by Prop.~\ref{propd:facts}\ref{enumd:facts:b}.
We can now inductively repeat the same arguments for~$y'$ and~$\varphi'$
that we have used for~$y$ and~$\varphi$.
Since~$y \nords y'$, we must eventually end up in \emph{Case~1.1.1.1}.
Hence the induction is  well-defined
and yields a graph chain~$\fchn_2$ for~$(y', \varphi', z, \psi)$.
Thus~$(x, \varphi), \fchn_1, \fchn_2$ is a graph chain for~$(x, \varphi, z, \psi)$.

\noindent\emph{Case~1.1.2: ($y$ does not exist)}
In this case we know that~$(z, \psi) \notin \prs(\varphi)$ in \detstsb($x$)
but must have been inserted in \filter($x, \prs$),
which is invoked at the end of \detstsb($x$).
Hence there exists a~$\chi \in \prsreach(\prs, x, \varphi)$ such that~$(z, \psi) \in \prs(\chi)$.
According to the definition of~$\prsreach$
there exist~$\varphi_0, \dotsc, \varphi_k \in \fmlev$
such that~$\chi = \varphi_k$ and~$(x, \varphi_0) \in \prs(\varphi)$
and~$(x, \varphi_{i+1}) \in \prs(\varphi_i)$ for all~$0 \leq i < k$.
Since~$(x, \varphi_0) \in \prs(\varphi)$ there exists an open child~$y_i$
such that~$(x, \varphi_0)$ is an element
of the corresponding~$\prsval_{\varphi, i}$ in \detstsb($x$).
Hence we can obtain a graph chain~$\fchn_{-1}$ for~$(x, \varphi, x, \varphi_0)$
exactly as in \emph{Case~1.1.1}.
Using the same arguments, we can also get
graph chains~$\fchn_i$ for~$(x, \varphi_i, x, \varphi_{i+1})$ for all~$0 \leq i < k$,
and a graph chain~$\fchn_k$ for~$(x, \chi, z, \psi)$.
Because we have~$\chi = \varphi_k$,
their concatenation~$\fchn_0, \dotsc, \fchn_k$ is a graph chain for~$(x, \varphi, z, \psi)$.

\noindent\emph{Case~1.2 ($\prs_x(\varphi) = \bot$):}
We distinguish whether or not there exists an open child~$y := y_i$ of~$x$
such that the corresponding set~$\prsval_{\varphi, i}$ in \detstsb($x$) is undefined.

\noindent\emph{Case~1.2.1: ($y$ exists)}
In this case \detstscint($x, y, \varphi$) must return undefined.
Hence we must have~$y \nords x$ by construction in \detstscint{}.
Our final case distinction is whether or not~$\prs_y(\varphi) = \bot$.

\noindent\emph{Case~1.2.1.1: ($\prs_y(\varphi) = \bot$)}
In this case then we can use the induction hypothesis on~$y$
to get a graph chain~$\fchn$ for~$(y, \varphi, z, \psi)$
for some~$z \in V$ and some~$\psi \in \fml \setminus \fmlev$.
Thus~$(x, \varphi), \fchn$ is a graph chain for~$(x, \varphi, z, \psi)$.

\noindent\emph{Case~1.2.1.2: ($\prs_y(\varphi) \not= \bot$)}
In this case there must exist a~$(y', \varphi') \in \prs_y(\varphi)$
such that \detstscint($x, y', \varphi'$) returns undefined.
Using the second part of the induction hypothesis on~$y$,
we obtain a graph chain~$\fchn_1$ for~$(y, \varphi, y', \varphi')$.
In particular~$y'$ is open
and we have~$y \nords y'$ by Prop.~\ref{propd:facts}\ref{enumd:facts:b}.
We can now inductively repeat the same arguments for~$y'$ and~$\varphi'$
that we have used for~$y$ and~$\varphi$.
Since~$y \nords y'$, we must eventually end up in \emph{Case~1.1.2.1}.
Hence the induction is  well-defined
and yields a graph chain~$\fchn_2$ for~$(y', \varphi', z, \psi)$
for some~$z \in V$ and some~$\psi \in \fml \setminus \fmlev$.
Thus~$(x, \varphi), \fchn_1, \fchn_2$ is a graph chain for~$(x, \varphi, z, \psi)$.

\noindent\emph{Case~1.2.2: ($y$ does not exist)}
In this case we know that~$\prs(\varphi)$ in \detstsb($x$) is defined.
Therefore~$\prs_x(\varphi)$ became undefined in \filter($x, \prs$),
which is invoked at the end of \detstsb($x$).
Hence there exists a~$\chi \in \prsreach(\prs, x, \varphi)$ such that~$\prs(\chi) = \bot$.
According to the definition of~$\prsreach$
there exist~$\varphi_0, \dotsc, \varphi_k \in \fmlev$
such that~$\chi = \varphi_k$ and~$(x, \varphi_0) \in \prs(\varphi)$
and~$(x, \varphi_{i+1}) \in \prs(\varphi_i)$ for all~$0 \leq i < k$.
Since~$(x, \varphi_0) \in \prs(\varphi)$, there exists an open child~$y := y_i'$
such that~$(x, \varphi_0)$ is an element
of the corresponding~$\prsval_{\varphi, i}$ in \detstsb($x, y_i'$).
Hence we can obtain a graph chain~$\fchn_{-1}$ for~$(x, \varphi, x, \varphi_0)$
exactly as in \emph{Case~1.1.1}.
Using the same arguments, we can also get
graph chains~$\fchn_i$ for~$(x, \varphi_i, x, \varphi_{i+1})$ for all~$0 \leq i < k$.
Since~$\prs(\chi) = \bot$, there exists an open child~$y := y_j'$
such that the corresponding~$\prsval_{\varphi, j}$ in \detstsb($x$) is undefined.
Thus we can obtain a graph chain~$\fchn_k$ for~$(x, \chi, z, \psi)$
for some~$z \in V$ and some~$\psi \in \fml \setminus \fmlev$ exactly as in \emph{Case~1.2.1}.
Because we have~$\chi = \varphi_k$,
the concatenation~$\fchn_0, \dotsc, \fchn_k$ is a graph chain for~$(x, \varphi, z, \psi)$.

\noindent\emph{Case~2 ($x$ is a state):}
By assumption~$\varphi$ is of the form~$\pea{l}{\chi}$
for some~$l \in \lprg$ and some~$\chi \in \fmlev$.
By distinguishing whether or not~$\prs_x(\pea{l}{\chi}) = \bot$,
we obtain the appropriate graph chain almost exactly as in \emph{Case~1}.
The main difference is that we have only one child of~$x$ to consider,
namely the successor of~$\pea{l}{\chi}$.

\noindent\emph{Case~3 ($x$ is a special node):}
First we inductively build a sequence~$\fchn_0$ as follows.
We start with~$\fchn_0 := (x, \varphi)$.
If~$\varphi$ is a \feal-formula then we stop.
Otherwise we extend~$\fchn_0$ with~$\varphi' := \ann_x(\varphi)$
which must be defined as~$x$ is a special node
and~$\ann_x$ is hence a full annotation.
Remember that we have~$\varphi \pzz \varphi'$.
If~$\varphi'$ is not an eventuality nor a \feal-formula then we stop.
Otherwise we extend~$\fchn_0$ with~$\varphi'' := \ann_x(\varphi')$ and so on.
The termination of this iteration is guaranteed by the fact
that~$x$ is open and hence~$\ann_x$ is non-cyclic.
It is not hard to see
that the final~$\fchn_0$ is of the form~$(x, \varphi), (x, \varphi'), \dotsc, (x, \chi)$
where~$\chi = \annchn(x, \varphi)$.
Hence we have~$\prs_x(\varphi) = \prs_x(\chi)$
by definition of \detstss{} and because of Prop.~\ref{propd:facts}\ref{enumd:facts:f}.
If~$\chi \notin \fmlev$ then we know~$\prs_x(\chi) = \prs_x(\varphi) = \bot$.
Furthermore~$\fchn_0$ is a graph chain for~$(x, \varphi, x, \chi)$ and we are done.
Otherwise~$\chi$ is of the form~$\pea{l}{\chi'}$
for some~$l \in \lprg$ and some~$\chi' \in \fmlev$.

By distinguishing whether or not~$\prs_x(\pea{l}{\chi'}) = \bot$,
we obtain the appropriate graph chain~$\fchn_1$ for~$(x, \pea{l}{\chi'}, z, \psi)$
exactly as in \emph{Case~1}.
Note in particular that we can use the induction hypothesis on the child of~$x$
which is a state since our eventuality is of the form~$\pea{l}{\chi'}$.
We can then conclude
that~$\fchn_0, \fchn_1$ is the appropriate graph chain for~$(x, \varphi, z, \psi)$.
\qed
\end{proof}

\begin{definition}
  For~$\varphi \in \fmlev$, let~$\extr(\varphi) \in \fml$ be the largest subformula of~$\varphi$
  such that~$\extr(\varphi) \notin \fmlev$ and~$\varphi \in \ppre(\extr(\varphi))$.
\end{definition}

It is easy to see that~$\extr(\varphi)$ is well-defined, in particular, that it is unique.

\begin{proposition}
  \label{propd:pzz}
  Let~$\psi_0, \dotsc, \psi_n$ be a finite sequence of formulae with~$n \geq 0$
  such that~$\psi_n \notin \fmlev$ and for every~$0 \leq i \leq n$:
  if~$\psi_i = \pea{l}{\chi}$ for some~$l \in \lprg$ and some~$\chi \in \fml$ then~$\psi_{i+1} = \chi$,
  else~$\psi_i \pzz \psi_{i+1}$.
  \begin{enumerate}
  \item\label{enumd:pzz:a} $\extr(\psi_j) = \psi_n$ for every~$0 \leq j \leq n$;
  \item\label{enumd:pzz:b} for every~$\chi \in \fml$
    such that~$\psi_n$ is a subformula of~$\chi$ and~$\chi$ is a subformula of~$\psi_0$,
    there exists a~$0 \leq j \leq n$ such that~$\chi = \psi_j$.
  \end{enumerate}
\end{proposition}

\begin{theorem}
  \label{theod:correctness}
  If root~$\rt \in V$ is open then~$\phi$ is satisfiable.
\end{theorem}
\begin{proof}
We will show that there exists a Hintikka structure for~$\phi$.
The theorem then follows with Prop.~\ref{propd:satisfiable}.

Let~$\gm$ (``p'' for pruned) be the subgraph of~$G$
that consists of all open nodes of~$G$ which are not states,
and all open states for which there exists a full non-cyclic annotation.
Note that open states which are children of open special nodes
are contained in~$\gm$ since the annotation of the special child
is trivially a full non-cyclic annotation for the state.
The edges of~$\gm$ are exactly the edges of~$G$ which connect two open nodes.
Then we use~$\gm$ to generate a structure~$H = (W, R, L)$ as described next.

Let~$W$ be the set of all states of~$\gm$.
To define~$R$ we first define the auxiliary function~$R: \lprg \to W \times W$.
For every~$l \in \lprg$ and every~$s, t \in W$, let~$(s, t) \in R'(b)$
iff~$\fset_s$ contains a formula~$\pea{l}{\varphi}$ for some formula~$\varphi \in \fml$
and there exists an $s$-path~$x_0 = s, x_1, \dotsc, x_{k+1}, x_{k+2} = t$ in~$\gm$
such that~$x_1$ is the successor of~$\pea{l}{\varphi}$,
each~$x_i, 1 \leq i \leq k$, is an $\alpha$- or a $\beta$-node or a special node,
and~$x_{k+1}$ is a special node.
Thus state~$t$ is a ``saturation'' of~$x_1$.
Then we set~$R_a := R'(a) \cup \{ (t,s) \mid (s, t) \in R'(\pcv{a}) \}$ for every~$a \in \aprg$.
Finally, we set~$L(s)$ to be~$\fset_s$ for all~$s \in W$.

As~$\rt$ is open, the only child~$y$ of~$\rt$ is also open
due to Lemma~\ref{lemmad:uptodate}\ref{enumd:uptodate:b}.
Moreover we have~$\phi \in \fset_y$ because~$y$ is the successor of~$\pea{d}{\phi} \in \fset_\rt$.
Thus there exists an open state~$s$ in~$G$, and thus in~$\gm$,
such that~$\phi \in \fset_s = L(s)$ due to Lemma~\ref{lemmad:uptodate}\ref{enumd:uptodate:c}
and Prop.~\ref{propd:facts}\ref{enumd:facts:d}\ref{enumd:facts:c}\ref{enumd:facts:a}.
Hence~$H$ is a structure for~$\phi$.

Next we show that~$H$ fulfils~H1 to~H6 and is thus a Hintikka structure for~$\phi$.
Let~$w \in W$, that is~$w$ is a state in~$\gm$.
In particular~$w$ is open.

\noindent{}H1:
The set~$L(w) = \fset_w$ cannot contain a contradiction as~$w$ is not closed.

\noindent{}H2 and H3 follow from the fact
that~$w$ is a state and not an $\alpha$- or a $\beta$-node.

\noindent{}H4 follows from the treatment of states in \expand{},
Lemma~\ref{lemmad:uptodate}\ref{enumd:uptodate:c}
as well as Prop.~\ref{propd:facts}\ref{enumd:facts:d}\ref{enumd:facts:c}\ref{enumd:facts:a}.

\noindent{}H5: Let~$\paa{l}{\varphi} \in L(w)$ and~$v \in W$ such that~$(w,v) \in R_b$.
We distinguish whether $(w,v) \in R'(b)$ or~$(v,w) \in R'(\pnnp{b})$.

If~$(w,v) \in R'(b)$, we have by definition of~$R'$
that~$\fset_w$ contains a formula~$\pea{l}{\psi}$ for some~$\psi \in \fml$
and there exists a path~$x_0 = w, x_1, \dotsc, x_{k+1}, x_{k+2} = v$ in~$\gm$
such that~$x_1$ is the successor of~$\pea{l}{\psi}$,
each~$x_i, 1 \leq i \leq k$, is an $\alpha$- or a $\beta$-node or a special node,
and~$x_{k+1}$ is a special node.
By construction we have~$\varphi \in \fset_{x_1}$,
and by Prop.~\ref{propd:facts}\ref{enumd:facts:c} and~\ref{enumd:facts:a}
we have~$\varphi \in \fset_v = L(v)$.

If~$(v,w) \in R'(\pnnp{b})$ then we have by definition of~$R'$
that the set~$\fset_v$ contains a formula~$\pea{\pnnp{b}}{\psi}$ for some~$\psi \in \fml$,
and that there exists a path $x_0 = v, x_1, \dotsc, x_{k+1}, x_{k+2} = w$ in~$\gm$
such that~$x_1$ is the successor of~$\pea{\pnnp{b}}{\psi}$,
and each~$x_i, 1 \leq i \leq k$, is an $\alpha$- or a $\beta$-node or a special node,
and~$x_{k+1}$ is a special node.
By construction and Prop.~\ref{propd:facts}\ref{enumd:facts:a}
we have~$\pst_{x_{k+1}} = v$ and~$\pap_{x_{k+1}} = \pnnp{b}$.
Together with Prop.~\ref{propd:facts}\ref{enumd:facts:g} and
the fact that~$\pnnp{(\pnnp{r})} = r$ and that~$x_{k+1}$ is not closed,
we can therefore conclude~$\varphi \in \fset_v = L(v)$.

\noindent{}H6: We use Lemma~\ref{lemmad:ex-fulfilling} as is shown next.
Let~$\varphi \in L(w) = \fset_w$ be an eventuality.
We first build a graph chain~$\fchn$ for~$(w, \varphi, y, \psi)$
for some~$y \in V$ and some~$\psi \in \fml \setminus \fmlev$
and then convert~$\fchn$ into a fulfilling chain.

By assumption there exists a full non-cyclic annotation~$\ann$ for~$w$.
First we inductively build a sequence~$\fchn_0$ as follows.
We start with~$\fchn_0 := (w, \varphi)$.
If~$\varphi$ is a \feal-formula then we stop.
Otherwise we extend~$\fchn_0$ with~$\varphi' := \ann(\varphi)$
which must be defined as~$\ann$ is a full annotation.
Remember that we have~$\varphi \pzz \varphi'$.
If~$\varphi'$ is not an eventuality nor a \feal-formula then we stop.
Otherwise we extend~$\fchn_0$ with~$\varphi'' := \ann(\varphi')$ and so on.
The termination of this iteration is guaranteed by the fact that~$\ann$ is non-cyclic.
Let the final~$\fchn_0$ be of the form~$(w, \varphi), (w, \varphi'), \dotsc, (w, \chi)$.
If~$\chi \notin \fmlev$ then~$\fchn_0$ is already a graph chain for~$(w, \varphi, w, \chi)$.
and we set~$\fchn := \fchn_0$.
Otherwise~$\chi$ is of the form~$\pea{l}{\chi'}$ for some~$l \in \lprg$ and some~$\chi' \in \fmlev$.
In this case we build another graph chain~$\fchn_1$ for~$(w, \chi, y, \psi)$
for some~$y \in V$ and some~$\psi \in \fml \setminus \fmlev$
by distinguishing whether or not~$\prs_w(\chi) = \bot$,
and then set~$\fchn := \fchn_0, \fchn_1$.

\emph{Case~1:} If~$\prs_w(\chi) = \bot$ then Lemma~\ref{lemmad:ex-fulfilling} provides us
with a graph chain~$\fchn_1$ for~$(w, \chi, y, \psi)$ for some~$y \in V$ and some~$\psi \in \fml \setminus \fmlev$.

\emph{Case~2:} If~$\prs_w(\varphi) \not= \bot$ then it contains at least one pair~$(y, \psi)$
such that~$w \nords y$ because of Prop.~\ref{propd:facts}\ref{enumd:facts:b}
and the fact that Rule~4 is not applicable.
Hence Lemma~\ref{lemmad:ex-fulfilling} gives us a graph chain for~$(w, \chi, y, \psi)$.
In particular~$y$ is open.
As the $\nords$-relation is well-founded,
we can inductively apply Lemma~\ref{lemmad:ex-fulfilling} on~$y$ and~$\psi$ and so on,
and must eventually encounter \emph{Case~1}.
Joining all sequences in the obvious way gives us a graph chain~$\fchn$ for~$(w, \chi, y, \psi)$
for some~$y \in V$ and some~$\psi \in \fml \setminus \fmlev$.

To convert~$\fchn$ into a fulfilling chain, we extend~$\fchn$ as follows.
Let~$\pi = z_0, \dotsc, z_k$ be a finite $y$-path in~$\gm$
such that~$t := z_k$ is the only state in~$\pi$ and~$\psi \in \fset_{z_i}$ for all~$0 \leq i \leq k$.
The existence of~$\pi$ follows from
Lemma~\ref{lemmad:uptodate}\ref{enumd:uptodate:c}
and Prop.~\ref{propd:facts}\ref{enumd:facts:d}\ref{enumd:facts:c}\ref{enumd:facts:a}.
Note that~$\pi = y$ iff~$y$ is a state.
We extend~$\fchn$ by the sequence~$(z_1, \psi), \dotsc, (z_k, \psi)$.
Now~$\fchn$ is a graph chain for~$(w, \varphi, t, \psi)$.

Finally, we convert~$\fchn$ into a fulfilling chain as follows.
Let~$\fchn$ be of the form~$(y_0, \psi_0), \dotsc, (y_n, \psi_n)$.
Next we replace each~$y_i, 0 \leq i \leq n$ in~$\fchn$ with the first state
which appears on the path $y_i, \dotsc, y_n$.
Furthermore we contract all consecutive repetitions of pairs.
Let the resulting~$\fchn$ be of the form~$(w_0, \psi_0'), \dotsc, (w_m, \psi_m')$.

It is not too hard to check that~$\fchn$ is a fulfilling chain for~$(H, w, \varphi, \psi)$.
Moreover, Prop.~\ref{propd:pzz}\ref{enumd:pzz:a} tells us that~$\psi = \extr(\varphi)$.
If~$\varphi = \pea{\prp{\gamma}}{\chi}$ for some~$\gamma \in \prg$ and some~$\chi \in \fml$
then~$\extr(\varphi)$ is a subformula of~$\chi$ which is obviously a subformula of~$\varphi$.
According to Prop.~\ref{propd:pzz}\ref{enumd:pzz:b},
there therefore exists a~$j \in \{ 0, \dotsc, m \}$ such that~$\chi = \psi_j'$.
Thus~$(w_0, \psi_0'), \dotsc, (w_j, \psi_j')$ is a fulfilling chain for~$(H, w, \pea{\prp{\gamma}}{\chi}, \chi)$.
Hence~H6 holds which concludes the proof.
\qed
\end{proof}

We next define some concepts related to models
and state some propositions which we will need in the remaining proofs.
\begin{definition}
  Let~$M = (W, R, \evv)$ be a model, $w \in W$, and~$\varphi, \psi \in \fml$.
  A \emph{model chain for~$(M, w, \varphi, \psi)$} is a finite sequence~$(w_0, \psi_0), \dotsc, (w_n, \psi_n)$
  of world-formula pairs with~$n \geq 0$ such that:
  \begin{enumerate}
  \item $w_0 = w$, $\psi_0 = \varphi$, $\psi_n = \psi$, and~$\psr{M, w_i}{\psi_i}$ for all~$0 \leq i \leq n$
  \item for all~$0 \leq i < n$:
    if~$\psi_i = \pea{l}{\chi}$ for some~$l \in \lprg$ and some~$\chi \in \fml$
    then~$\psi_{i+1} = \chi$ and~$\prel{w_i}{R_b}{w_{i+1}}$, else~$\psi_i \pzz \psi_{i+1}$ and~$w_i = w_{i+1}$.
  \end{enumerate}
\end{definition}

\begin{proposition}
  \label{propd:chain}
  Let~$M = (W, R, \evv)$ be a model, $w \in W$,
  and~$\varphi, \psi \in \fml$ such that~$\varphi \in \ppre(\psi)$ and~$\psr{M, w}{\varphi}$.
  Then there exists a model chain for~$(M, w, \varphi, \psi)$.
\end{proposition}

\begin{definition}
  Let~$X$ and~$Y$ be sets and~$f, g: X \to \opt{Y}$ be functions.
  We say that~$g$ \emph{extends}~$f$ iff $\forall x \in X.\: f(x) \not= \bot \ximp g(x) = f(x)$.

  Let~$\seq{X}$ be the set of non-empty, finite sequences of elements in~$X$.
  For a sequence~$\fchn = x_0, \dotsc, x_n$ in~$\seq{X}$,
  let~$|\fchn| := n+1$ be the length of~$\fchn$.
  Furthermore, for all~$0 \leq i < |\fchn|$,
  we define~$\fchn_i := x_i$ and~$\fchn_{\geq i} := x_i, x_{i+1}, \dotsc, x_n$.
\end{definition}

\begin{definition}
  Let~$M = (W, R, \evv)$ be a model, $w \in W$, and~$\fset \subseteq \fml$.
  We write~$\psr{M, w}{\fset}$ iff~$\psr{M, w}{\varphi}$ for all~$\varphi \in \fset$.
  Furthermore, let~$\ann: \fmlev \to \opt{\fml}$ be an annotation for~$\fset$
  and~$\mchn: \fmlev \to \opt{\big( \seq{W \times \fml} \big)}$ a function.
  We say that~$(M, w, \mchn)$ \emph{annotated-satisfies}~$(\fset, \ann)$
  iff the following conditions hold:
  \begin{itemize}
  \item $\psr{M, w}{\fset}$
  \item $\mchn(\varphi)$ is defined for all~$\varphi \in \fset \cap \fmlev$
  \item for all~$\varphi \in \fmlev$ such that~$\fchn := \mchn(\varphi)$ is defined, we have:
    \begin{itemize}
    \item the sequence~$\fchn$ is a model chain for~$(M, w, \varphi, \extr(\varphi))$
    \item for all~$0 \leq i < |\fchn|$,
      if~$\fchn_i$ is of the form~$(w, \psi)$ for some~$\psi \in \fmlev$
      then~$\bot \not= \mchn(\psi) = \fchn_{\geq i}$
    \item if~$\ann(\varphi) \not= \bot$ then~$\fchn_1 = (w, \ann(\varphi))$.
    \end{itemize}
  \end{itemize}
  If there exists such a triple~$(M, w, \mchn)$
  we call~$(\fset, \ann)$ \emph{annotated-satisfiable}.
\end{definition}

\begin{proposition}
  \label{propd:annsat}
  Let~$M = (W, R, \evv)$ be a model, $w \in W$, $\varphi \in \fml$, and~$\fset \subseteq \fml$ be a finite set,
  such that~$\psr{M, w}{\varphi}$ and~$\psr{M, w}{\fset}$.
  \begin{enumerate}
  \item\label{enumd:annsat:a} Let~$\ann: \fmlev \to \opt{\fml}$ be an annotation for~$\fset$
    and~$\mchn: \fmlev \to \opt{\big( \seq{W \times \fml} \big)}$ a function
    such that~$(M, w, \mchn)$ annotated-satisfies~$(\fset, \ann)$.
    If we have~$\varphi \notin \fmlev$ or~$\ann(\varphi) = \bot$
    then there exists an extension~$\mchn'$ of~$\mchn$
    such that~$(M, w, \mchn')$ annotated-satisfies~$(\fset \cup \{ \varphi \}, \ann)$.
  \item\label{enumd:annsat:b} Let~$\fchn$ be a model chain for~$(M, w, \varphi, \extr(\varphi))$
    such that~$\fchn_i \not= \fchn_j$ for all~$0 \leq i, j < |\fchn|, i \not= j$.
    Then there exists a function~$\mchn': \fmlev \to \opt{\big( \seq{W \times \fml} \big)}$
    such that~$(M, w, \mchn')$ annotated-satisfies~$(\fset \cup \{ \varphi \}, \eann)$
    and~$\mchn'(\varphi) = \fchn$ if~$\varphi \in \fmlev$.
  \end{enumerate}
\end{proposition}

For a node~$x \in V$, we say that~$(M, w, \mchn)$ annotated-satisfies~$x$
iff~$(M, w, \mchn)$ annotated-satisfies~$(\fset_x, \ann_x)$.
Similarly, we say that~$x$ is annotated-satisfiable
iff~$(\fset_x, \ann_x)$ is annotated-satisfiable.

\begin{definition}
  Let~$M = (W, R, \evv)$ be a model, $w, v \in W$,
  and~$\mchn: \fmlev \to \opt{\big( \seq{W \times \fml} \big)}$ be a function.
  \begin{itemize}
  \item Let~$x \in V$ be a state.
    We say that~$(M, w)$ \emph{realises}~$x$
    iff~$\psr{M, w}{\fset_x}$ and~$\npsr{M, w}{(\fset_x \cup \fset)}$ for all~$\fset \in \alt_x$.
  \item Let~$x$ be an $\alpha$- or a $\beta$-node or a special node.
    We say that~$(M, v, w, \mchn)$ \emph{realises}~$x$
    iff~$\psr{M, v}{\fset_{\pst_x}}$,
    $\npsr{M, v}{(\fset_{\pst_x} \cup \fset)}$ for all~$\fset \in \alt_x$,
    $\prel{v}{R_{\pap_x}}{w}$
    and~$(M, w, \mchn)$ annotated-satisfies~$x$.
  \end{itemize}
  We call a node~$x \in V$ \emph{realisable}
  iff there exists an appropriate tuple~$(M, w)$ or~$(M, v, w, \mchn)$ which realises~$x$.
\end{definition}  

\begin{lemma}
  \label{lemmad:annsatprop}
  Let~$M = (W, R, \evv)$ be a model, $w, v \in W$,
  and~$\mchn: \fmlev \to \opt{\big( \seq{W \times \fml} \big)}$ be a function,
  \begin{enumerate}
  \item\label{enumd:annsatprop:a}
    Let~$x \in V$ be an $\alpha$- or a $\beta$-node such that~$x$ is up-to-date
    and~$(M, v, w, \mchn)$ realises~$x$.
    Then there exists a child~$y$ of~$x$ and an extension~$\mchn'$ of~$\mchn$
    such that~$(M, v, w, \mchn')$ realises~$y$.
  \item\label{enumd:annsatprop:b}
    Let~$x \in V$ be a state such that~$(M, w)$ realises~$x$.
    Moreover, let~$\pea{l}{\psi} \in \fset_x$
    and~$y \in V$ be the successor of~$\pea{l}{\psi}$ such that~$\alt_y \subseteq \alt_x$.
    Furthermore, let~$\fchn$ be a model chain for~$(M, w, \pea{l}{\psi}, \extr(\pea{l}{\psi}))$
    such that~$\fchn_i \not= \fchn_j$ for all~$0 \leq i, j < |\fchn|, i \not= j$,
    and let~$v' \in W$ such that~$\fchn_1 = (v', \psi)$.
    Then there exists a function~$\mchn': \fmlev \to \opt{\big( \seq{W \times \fml} \big)}$
    such that~$(M, w, v', \mchn')$ realises~$y$
    and~$\mchn'(\psi) = \fchn_{\geq 1}$ if~$\psi \in \fmlev$.
  \item\label{enumd:annsatprop:c}
    Let~$x \in V$ be a special node such that~$x$ is up-to-date
    and~$(M, v, w, \mchn)$ realises~$x$.
    Then there either exists a child~$y$ of~$x$ and an extension~$\mchn'$ of~$\mchn$
    such that~$y$ is not a state and~$(M, v, w, \mchn')$ realises~$y$,
    or~$(M, w)$ realises~$\chld(x, \mrk)$.
  \end{enumerate}
\end{lemma}
\begin{proof}
\noindent\ref{enumd:annsatprop:a}:
We distinguish whether or not the principal formula~$\varphi \in \fset_x$,
which is decomposed in~$x$, is an eventuality.

If~$\varphi \in \fmlev$ then~$\fchn := \mchn(\varphi)$ is a model chain for~$(M, w, \varphi, \extr(\varphi))$.
Since~$\varphi$ is not a \feal-formula,
we therefore have~$\fchn_1 = (w, \varphi')$ for some~$\varphi' \in \fml$
such that~$\varphi \pzz \varphi'$ and~$\psr{M, w}{\varphi'}$.
Due to the definition of~$\pzz$ and \expand{}, there exists a child~$y \in V$ of~$x$
such that~$\fset_y = \fset_x \uplus \{ \varphi' \}$ and~$\ann_y = \updfkt{\ann_x}{\varphi}{\varphi'}$.
Together with the fact that~$(M, w, \mchn)$ annotated-satisfies~$x$,
is is not too hard to see that~$(M, w, \mchn)$ also annotated-satisfies~$y$.
Since~$x$ is up-to-date and by definition of \detstsb{}, we have~$\alt_y \subseteq \alt_x$.
Together with Prop.~\ref{propd:facts}\ref{enumd:facts:a}
and the assumption that~$(M, v, w, \mchn)$ realises~$x$,
this implies that~$(M, v, w, \mchn)$ realises~$y$.

If~$\varphi \notin \fmlev$, let~$y$ be a child of~$x$ such that~$\psr{M, w}{\fset_y}$.
The existence of~$y$ follows from~$\psr{M, w}{\fset_x}$ and Prop.~\ref{propd:axioms}.
Furthermore, let~$\varphi_1, \dotsc, \varphi_k (1 \leq k \leq 2)$ be all formulae
which were actually added to~$\fset_x$ to create~$\fset_y$ in \expand{}.
If~$\varphi_i (1 \leq k \leq 2)$ is an eventuality,
we have~$\ann_y(\varphi_i) = \bot$ because of~$\ann_x = \ann_y$ and~$\varphi_i \notin \fset_x$.
We can then conclude that~$(M, w, \mchn')$ annotated-satisfies~$y$
for some extension~$\mchn'$ of~$\mchn$
by repeatedly applying Prop.~\ref{propd:annsat}\ref{enumd:annsat:a}.
Since~$x$ is up-to-date and by definition of \detstsb{}, we have~$\alt_y \subseteq \alt_x$.
Together with Prop.~\ref{propd:facts}\ref{enumd:facts:a}
and the assumption that~$(M, v, w, \mchn)$ realises~$x$,
this implies that~$(M, v, w, \mchn')$ realises~$y$.

\noindent\ref{enumd:annsatprop:b}:
It is easy to see that~$\fchn_{\geq 1}$ is a model chain for~$\psi$
and that we have~$\prel{w}{R_b}{v'}$ and~$\psr{M, v'}{\psi}$.
By construction in \expand{}, we have~$\ann_y = \eann$
and~$\fset_y = \{ \psi \} \cup \fset$ for some~$\fset \subseteq \fml$ such that~$\paa{l}{\fset} \subseteq \fset_x$.
Since~$(M, w)$ realises~$x$,
we have~$\psr{M, w}{\paa{l}{\fset}}$ and thus~$\psr{M, v'}{\fset}$.
We can thus apply Prop.~\ref{propd:annsat}\ref{enumd:annsat:b}
on~$M$, $v'$, $\psi$, $\fset$, and~$\fchn_{\geq 1}$
and obtain a function~$\mchn': \fmlev \to \opt{\big( \seq{W \times \fml} \big)}$
such that~$(M, v', \mchn')$ annotated-satisfies~$(\fset \cup \{ \psi \}, \eann)$, and hence~$y$.
By definition of \detstsb{} and the assumptions that~$(M, w)$ realises~$x$ and~$\alt_y \subseteq \alt_x$,
this implies that~$(M, w, v', \mchn')$ realises~$y$
and~$\mchn'(\psi) = \fchn_{\geq 1}$ if~$\psi \in \fmlev$.

\noindent\ref{enumd:annsatprop:c}:
Assume that~$(M, w)$ does not realise~$y := \chld(x, \mrk)$.
As~$(M, v, w, \mchn)$ realises~$x$, we have~$\psr{M, w}{\fset_x}$.
Furthermore, we have~$\fset_x = \fset_y$ by construction in \expand{}.
Hence there must exist a~$\Gamma \in \alt_x$ such that~$\psr{M, w}{(\fset_y \cup \Gamma)}$.
By construction in \expand{}, we have~$\fset_y = \fset_x$.
Moreover, by construction in \detstsd{} and since~$x$ is up-to-date,
the node~$x$ has a child~$z$ which is not a state
such that~$\fset_z = \fset_x \cup \Gamma$, $\ann_x = \ann_z$, $\pst_x = \pst_z$, $\pap_x = \pap_z$,
and~$\alt_z \subseteq \alt_x$.
In particular,
we have~$\ann_z(\chi) = \bot$ for all eventualities in~$\Gamma$.
We can then conclude that~$(M, w, \mchn')$ annotated-satisfies~$z$
for some extension~$\mchn'$ of~$\mchn$
by repeatedly applying Prop.~\ref{propd:annsat}\ref{enumd:annsat:a}.
Together with the assumption that~$(M, v, w, \mchn)$ realises~$x$,
this implies that~$(M, v, w, \mchn')$ realises~$z$.
\qed
\end{proof}

\begin{definition}
  For a node~$x \in V$ and an eventuality~$\varphi \in \fset_y \cap \fmlev$,
  we define the property~$\ptya(x, \varphi)$ as
  ``$x \text{ is open } \xand \prs_x(\varphi) \not= \bot \xand
  \forall (y, \psi) \in \prs_x(\varphi).\: \big(\sts_y \not= \closed(\cdot) \ximp \ptya(y, \psi) \big)$''.
\end{definition}
Note that although~$\ptya(x, \varphi)$ is recursive,
it is well-defined because of Prop.~\ref{propd:facts}\ref{enumd:facts:b}.

\begin{lemma}
  \label{lemmad:property}
  In this lemma, let~$G$ be the graph at any time between two rule applications
  when Rule~3 is not applicable.
  Let~$x \in V$ and~$\varphi \in \fset_x \cap \fmlev$ such that~$\ptya(x, \varphi)$ holds.
  Furthermore, let~$y$ be a child of~$x$ which is not closed.
  If~$x$ is a state or a special node, we require additionally
  that~$\varphi = \pea{l}{\chi}$ for some~$l \in \lprg$ and some~$\chi \in \fmlev$.
  If~$x$ is a state, we also require that~$y$ is the successor of~$\varphi$.
  If~$x$ is a state we define~$\psi := \chi$, else we define~$\psi := \varphi$.
  In all cases we have~$\ptya(y, \psi)$.
\end{lemma}
\begin{proof}
We assume that~$\ptya(y, \psi)$ does not hold and derive a contradiction.
Because of Prop.~\ref{propd:facts}\ref{enumd:facts:b},
this means that there exists a finite sequence~$(y_0, \psi_0), \dotsc, (y_n, \psi_n)$
of node-eventuality pairs such that:
\begin{itemize}
\item $y_0 = y$, $\psi_0 = \psi$, and~$y_n$ is not closed;
\item $y_i$ is open and $(y_{i+1}, \psi_{i+1}) \in \prs_{y_i}(\psi_i) \not= \bot$ for all~$0 \leq i < n$;
\item either~$y_n$ is not open or~$\prs_{y_n}(\psi_n) = \bot$.
\end{itemize}
Note that~$n = 0$ is possible
and that~$\ptya(y_i, \psi_i)$ does not hold for all~$0 \leq i \leq n$.

Let~$m \in {0, \dotsc, n}$ be the smallest index such that not~$y_m \nords x$;
or~$m := n$ if~$y_i \nords x$ for all~$0 \leq i \leq n$.
Then the sequence~$(y_0, \psi_0), \dotsc, (y_m, \psi_m)$ has the following properties,
most of them are inherited from above:
\begin{itemize}
\item $y_0 = y$, $\psi_0 = \psi$, and~$y_m$ is not closed;
\item $y_i$ is open and $(y_{i+1}, \psi_{i+1}) \in \prs_{y_i}(\psi_i) \not= \bot$ for all~$0 \leq i < n$;
\item either~$y_n$ is not open or~$\prs_{y_n}(\psi_n) = \bot$ or not~$y_m \nords x$.
\end{itemize}
The following arguments rely on the fact that all nodes are up-to-date and
on the definitions of \detstscint{} and
-- depending on whether~$x$ is an $\alpha$/$\beta$-node, a state, or a special node --
\detstsb{} or \detstss{} or \detstsd{}, respectively.

If~$\prs_{y_m}(\psi_m) = \bot$, it is not too hard to see
that $\detstscint(x, \varphi) = \bot$ and hence~$\prs_x(\varphi) = \bot$
which contradicts the assumption that~$\ptya(x, \varphi)$ holds.

If~$y_n$ is not open or not~$y_m \nords x$,
it is not too hard to see
that we have $(y_m, \psi_m) \in \detstscint(x, \varphi)$.
Together with the definition of \filter{},
this implies~$(y_m, \psi_m) \in \prs_x(\varphi) \cup \{ (x, \chi) \mid \chi \in \fset_x \cap \fmlev \}$.
If~$(y_m, \psi_m) \in \prs_x(\varphi) \cup \{ (x, \varphi) \}$
then~$\ptya(y_m, \psi_m)$ follows from~$\ptya(x, \varphi)$,
but this contradicts the fact that~$\ptya(y_m, \psi_m)$ does not hold.
If~$(y_m, \psi_m) = (x, \chi)$ for some~$\chi \in \fset_x \cap \fmlev$ with~$\chi \not= \varphi$,
we have~$\bot \not= \prs_x(\chi) \subseteq \prs_x(\varphi)$
by construction in \filter{}
and because of the transitivity in the function~$\prsreach$.
Hence~$\ptya(x, \varphi)$ implies~$\ptya(x, \chi) = \ptya(y_m, \psi_m)$
which again is a contradiction.
\qed
\end{proof}

\begin{lemma}
  \label{lemmad:realise}
  If a node in~$G$ is closed then it is not realisable.
\end{lemma}
\begin{proof}
We proceed by induction on the order in which the nodes are closed.
That is, when dealing with a node which has just become closed,
we can assume that all other nodes in~$G$ which are already closed are not realisable.
We must consider all cases in the algorithm where a node can be closed.

If a node is closed because it contains an immediate contradiction,
it cannot be realised by definition.
If a node is closed because it contains an ``at a world'' cycle,
it cannot be a state by construction.
Furthermore, it is easy to see 
that nodes containing an ``at a world'' cycle cannot be annotated-satisfiable.

If an $\alpha$- or a $\beta$-node or a special node~$x$ is closed
because all its children are closed then~$x$ is not realisable
because of the induction hypothesis on the children
and Lemma~\ref{lemmad:annsatprop}\ref{enumd:annsatprop:a} and~\ref{enumd:annsatprop:c}.

If a state~$x \in V$ is closed because one of its children is closed,
let~$\pea{l}{\varphi}$ be the \feal-formula in~$x$
such that its successor~$y$ is already closed and~$\alt_x = \alt_y$.
The claim then follows from the induction hypothesis on~$y$, Prop.~\ref{propd:chain},
and Lemma~\ref{lemmad:annsatprop}\ref{enumd:annsatprop:b}.
Note that a model chain can always be shortened
such that it does not contain a pair twice.

The last and most interesting case is
when a node~$\ini{x} \in V$ is closed in Rule~4
because there exists an eventuality~$\ini{\varphi} \in \fset_{\ini{x}} \cap \fmlev$
such that~$\prs_{\ini{x}}(\ini{\varphi}) = \emptyset$.
For the rest of the proof,
we consider~$G$ at that moment right before closing~$\ini{x}$.
Next we assume that~$\ini{x}$ is realisable and derive a contradiction.

For a contradiction, we distinguish whether or not~$\ini{x}$ is a state.
If~$\ini{x}$ is a state, let~$M = (W, R, \evv)$ be a model and~$\ini{w} \in W$
such that~$(M, \ini{w})$ realises~$\ini{x}$.
Furthermore, let~$\fchn$ be a model chain for~$(M, \ini{w}, \ini{\varphi}, \extr(\ini{\varphi}))$
which exists due to Prop.~\ref{propd:chain}.
If~$\ini{x}$ is an $\alpha$- or a $\beta$-node or a special node,
let~$M = (W, R, \evv)$ be a model, $\ini{w}, \ini{v} \in W$,
and~$\ini{\mchn}: \fmlev \to \opt{\big( \seq{W \times \fml} \big)}$ a function
such that~$(M, \ini{v}, \ini{w}, \ini{\mchn})$ realises~$\ini{x}$.
Furthermore, let~$\fchn := \ini{\mchn}(\ini{\varphi})$,
that is~$\fchn$ is a model chain for~$(M, \ini{w}, \ini{\varphi}, \extr(\ini{\varphi}))$.
In both cases, it is not too hard to see that we can assume without loss of generality
that~$\fchn_i \not= \fchn_j$ for all~$0 \leq i, j < |\fchn|, i \not= j$.

We will inductively define an arbitrarily long sequence~$\nchn \in \seq{V \times \Nat}$,
initially starting with~$\nchn := (\ini{x}, 0)$,
such that the following invariant is maintained for all~$0 \leq k < |\nchn|$:
\begin{description}
\item[Invariant:] Let~$(x, i) := \nchn_k \in V \times \Nat$.
  \begin{itemize}
  \item if~$k > 0$ and~$(y, j) := \nchn_{k-1}$ then:
    $x$ is a child of~$y$ and~$i \geq j$, and if~$z$ is a state then~$i > j$;
  \item we have~$0 \leq i < |\fchn|$, let~$(w, \varphi) := \fchn_i \in W \times \fml$;
  \item $\varphi \in \fset_y \cap \fmlev$ and~$\ptya(x, \varphi)$ holds;
  \item if~$x$ is a state then~$(M, w)$ realises~$x$ and~$\varphi$ is a \feal-formula;
  \item if~$x$ is an $\alpha$- or a $\beta$-node or a special node,
    there exists a~$v \in W$ and a function~$\mchn: \fmlev \to \opt{\big( \seq{W \times \fml} \big)}$
    such that~$\mchn(\varphi) = \fchn_{\geq i}$ and~$(M, v, w, \mchn)$ realises~$x$.
  \end{itemize}
\end{description}
It is not difficult to see
that the initial sequence~$\nchn = (\ini{x}, 0)$ fulfils the invariant.
Note in particular that if~$\ini{x}$ is a state
then~$\ini{\varphi}$ is a \feal-formula by construction in \detstss{}.

Before we describe the construction of~$\nchn$,
we show how we can use~$\nchn$ to derive a contradiction.
Because of Prop.~\ref{propd:facts}\ref{enumd:facts:d},
we can make~$\nchn$ long enough
so that it contains~$|\fchn| + 1$ (not necessarily different) states.
Hence~$\nchn$ must contain an element~$(x, j)$ with~$j \geq |\fchn|$ due to the invariant;
but this is not possible since the invariant also guarantees
that the natural numbers in~$\nchn$ are strictly smaller than~$|\fchn|$.
Next we will show how to construct~$\nchn$.

Let~$(x, i)$ be the last pair in~$\nchn$ as constructed so far and~$(w, \varphi) := \fchn_i$.
We distinguish whether~$x$ is an $\alpha$/$\beta$-node or a state or a special node.

\noindent\emph{Case~1 ($x$ is $\alpha$/$\beta$-node):}
Let~$v \in W$ and~$\mchn: \fmlev \to \opt{\big( \seq{W \times \fml} \big)}$ be a function
such that~$\mchn(\varphi) = \fchn_{\geq i}$ and~$(M, v, w, \mchn)$ realises~$x$.
Using Lemma~\ref{lemmad:annsatprop}\ref{enumd:annsatprop:a},
we obtain a child~$y$ of~$x$ and an extension~$\mchn'$ of~$\mchn$
such that~$(M, v, w, \mchn')$ realises~$y$.
In particular, we have~$\mchn'(\varphi) = \fchn_{\geq i}$.
Furthermore~$y$ cannot be closed because of the induction hypothesis.
We extend~$\nchn$ by~$(y, i)$.
It remains to show that property~$\ptya(y, \varphi)$ holds,
but this is exactly what Lemma~\ref{lemmad:property} does.

\noindent\emph{Case~2 ($x$ is a state):}
Because of the invariant, we know that~$\varphi$ is of the form~$\pea{l}{\psi} \in \fmlev$
for some~$l \in \lprg$ and some~$\psi \in \fmlev$.
Let~$v' \in W$ such that~$(v', \psi) = \mchn(\pea{l}{\psi})_1 = \fchn_{i+1}$
and~$y$ be the successor of~$\pea{l}{\psi}$.
By construction in \detstss{} and the fact that~$x$ is open and up-to-date,
we have~$\alt_y \subseteq \alt_x$.
Using Lemma~\ref{lemmad:annsatprop}\ref{enumd:annsatprop:b}
on~$y$ and~$\pea{l}{\psi}$ and~$\fchn{\geq i}$,
we obtain a function~$\mchn': \fmlev \to \opt{\big( \seq{W \times \fml} \big)}$
such that~$(M, w, v', \mchn')$ realises~$y$ and~$\mchn'(\psi) = \fchn_{\geq i+1}$.
In particular we have~$i+1 < |\fchn|$ since~$\psi \in \fmlev$.
Furthermore~$y$ cannot be closed,
either because of~$y = \ini{x}$ or because of the induction hypothesis.
We extend~$\nchn$ by~$(y, i+1)$.
It remains to show that property~$\ptya(y, \psi)$ holds
which is done in Lemma~\ref{lemmad:property}.

\noindent\emph{Case~3 ($x$ is a special node):}
Let~$v \in W$ and~$\mchn: \fmlev \to \opt{\big( \seq{W \times \fml} \big)}$ be a function
such that~$\mchn(\varphi) = \fchn_{\geq i}$ and~$(M, v, w, \mchn)$ realises~$x$.
Using Lemma~\ref{lemmad:annsatprop}\ref{enumd:annsatprop:c},
either~$(M, w)$ realises~$y := \chld(x, \mrk)$
or there exists a child~$y$ of~$x$ and an extension~$\mchn'$ of~$\mchn$
such that~$y$ is not a state and~$(M, v, w, \mchn')$ realises~$y$.
In both cases~$y$ cannot be closed because of the induction hypothesis.

If~$\varphi$ is a not a \feal-formula then we consider~$\varphi' := \ann_x(\varphi)$
which must be defined as~$x$ is a special node.
We know~$\varphi' \in \fmlev$
since we would have~$\prs_x(\varphi) = \bot$ by construction in \detstsb{} otherwise.
As~$(M, w, \mchn)$ annotated-satisfies~$x$,
we therefore have~$\mchn(\varphi)_1 = (w, \varphi')$ and hence~$\mchn(\varphi') = \mchn(\varphi)_{\geq 1}$.
Due to the invariant we have~$\mchn(\varphi) = \fchn_{\geq i}$
and thus~$\mchn(\varphi') = \fchn_{\geq i+1}$ and in particular~$i+1 < |\fchn|$.
Furthermore we have~$\prs_x(\varphi) = \prs_x(\varphi')$
by definition of \detstss{} and because of Prop.~\ref{propd:facts}\ref{enumd:facts:f}.
Thus~$\ptya(x, \varphi')$ holds since~$\ptya(x, \varphi)$ holds.
If~$\varphi'$ is a not a \feal-formula, we consider~$\varphi'' := \ann_x(\varphi')$ and so on.
Since~$x$ is open, and hence~$\ann_x$ is non-cyclic,
we will eventually obtain a \feal-formula~$\pea{l}{\psi} \in \fmlev$
for some~$l \in \lprg$ and some~$\psi \in \fml$ such that~$\ptya(x, \pea{l}{\psi})$ holds
and~$\mchn(\pea{l}{\psi}) = \fchn_{\geq i+j}$ for some~$j \in \Nat$ with~$i+j < |\fchn|$.
Since~$\mchn'$, if it is needed, is an extension of~$\mchn$,
we also have~$\mchn'(\pea{l}{\psi}) = \fchn_{\geq i+j}$.
In particular, we have~$\fchn_{\geq i+j} = (w, \pea{l}{\psi})$.
Note that~$\varphi = \pea{l}{\psi}$, and thus~$j = 0$, is possible.

We extend~$\nchn$ by~$(y, i + j)$.
It remains to show that property~$\ptya(y, \pea{l}{\psi})$ holds
which is shown by Lemma~\ref{lemmad:property} and the established fact
that~$\ptya(x, \pea{l}{\psi})$ holds.
\qed
\end{proof}

\begin{theorem}
  \label{theod:completeness}
  If~$\rt$ is closed then~$\phi$ is unsatisfiable.
\end{theorem}
\begin{proof}
  Assume for a contradiction that~$\phi$ is satisfiable.
  Since the dummy atomic program~$d$ does not occur in~$\phi$,
  it is not too hard to see that~$\pea{d}{\phi}$ is also satisfiable.
  Hence there exists a model~$M = (W, R, \evv)$ and a world~$w \in W$
  such that~$\psr{M, w}{\pea{d}{\phi}}$.
  Together with Prop.~\ref{propd:facts}\ref{enumd:facts:h}
  and the fact that~$\fset_\rt = \{ \pea{d}{\phi} \}$,
  this implies that~$(M, w)$ realises~$x$
  which contradicts Lemma~\ref{lemmad:realise}.
\qed
\end{proof}

\end{document}